\pdfoutput=1
\documentclass[runningheads,UKenglish]{llncs}
\usepackage[T1]{fontenc}
\usepackage[utf8]{inputenc}

\usepackage{booktabs}
\usepackage{wrapfig}
\usepackage{amssymb,amsmath,stmaryrd,nicefrac,proof}
\usepackage{xypic}
\usepackage{tikz-cd}
\usepackage[hidelinks]{hyperref}

\usepackage[appendix=append,bibliography=common]{apxproof}

\usepackage{enumitem}

\newtheoremrep{theorem}{Theorem}
\newtheoremrep{proposition}[theorem]{Proposition}

\spnewtheorem{thm}[theorem]{Theorem}{\bfseries}{\itshape}
\spnewtheorem{clm}[theorem]{Claim}{\bfseries}{\itshape}
\spnewtheorem{cor}[theorem]{Corollary}{\bfseries}{\itshape}
\spnewtheorem{cnj}[theorem]{Conjecture}{\bfseries}{\itshape}
\spnewtheorem{lem}[theorem]{Lemma}{\bfseries}{\itshape}
\spnewtheorem{lemdefn}[theorem]{Lemma and Definition}{\bfseries}{\itshape}
\spnewtheorem{prop}[theorem]{Proposition}{\bfseries}{\itshape}
\spnewtheorem{defn}[theorem]{Definition}{\bfseries}{\upshape}
\spnewtheorem{rem}[theorem]{Remark}{\bfseries}{\upshape}
\spnewtheorem{notation}[theorem]{Notation}{\bfseries}{\upshape}
\spnewtheorem{expl}[theorem]{Example}{\bfseries}{\upshape}
\spnewtheorem{thmdefn}[theorem]{Theorem and Definition}{\bfseries}{\itshape}
\spnewtheorem{propdefn}[theorem]{Proposition and Definition}{\bfseries}{\itshape}
\spnewtheorem{assumption}[theorem]{Assumption}{\bfseries}{\upshape}
\spnewtheorem{algorithm}[theorem]{Algorithm}{\bfseries}{\upshape}

\usepackage[author=anonymous,marginclue,footnote]{fixme}

\FXRegisterAuthor{jf}{ajf}{JF}
\FXRegisterAuthor{pw}{apw}{PW}
\FXRegisterAuthor{ls}{als}{LS}
\FXRegisterAuthor{hb}{ahb}{HB}
\FXRegisterAuthor{bk}{abk}{BK}
\FXRegisterAuthor{sg}{asg}{SG}
\FXRegisterAuthor{km}{akm}{KM}

\authorrunning{J. Forster, L. Schr{\"o}der, P. Wild, H. Beohar, S. Gurke, K. Messing}
\author{Jonas Forster \inst{1} \and Lutz Schr\"oder \inst{1} \and Paul Wild \inst{1} \and Harsh Beohar \inst{2} \and Sebastian Gurke \inst{3} \and Karla Messing\inst{3}}

\institute{Friedrich-Alexander-Universität Erlangen-Nürnberg \and University of Sheffield \and Universität Duisburg-Essen}

\title{Graded Semantics and Graded Logics for Eilenberg-Moore Coalgebras}

\usepackage[colorinlistoftodos,textsize=tiny,color=orange!70
, disable%
]{todonotes}

\newcommand{\by}[1]{\qquad \text{(#1)}}

\newcommand{\power}{\mathcal P} 
\newcommand{\dist}{\mathcal D} 
\newcommand{\id}{\mathit{id}} 
\newcommand{\Id}{\mathit{Id}} 

\newcommand{\V}{\mathcal{V}}

\newcommand{\pmet}{\mathbf{PMet}}
\newcommand{\dpmet}{\mathbf{DPMet}}
\newcommand{\met}{\mathbf{Met}}

\newcommand{\Ccat}{\mathbf{C}}
\newcommand{\alg}[1]{\mathbf{EM}(#1)}

\newcommand{\Set}{\mathbf{Set}}
\newcommand{\modal}[1]{\langle #1 \rangle}

\makeatletter
\newcommand{\eqnum}{\refstepcounter{equation}\textup{\tagform@{\theequation}}}
\makeatother

\newcommand{\Galg}[2]{\text{Alg}_#1(#2)}
\newcommand{\mbar}{\overline{M}_1}
\newcommand{\sem}[1]{\llbracket #1 \rrbracket}
\newcommand{\ev}[1]{\text{ev}_{#1}}

\newcommand{\nat}{\mathbb{N}}

\begin{document}

\maketitle

\begin{abstract}
  Coalgebra, as the abstract study of state-based systems, comes
  naturally equipped with a notion of behavioural equivalence that
  identifies states exhibiting the same behaviour. In many cases,
  however, this equivalence is finer than the intended
  semantics. Particularly in automata theory, behavioural equivalence
  of nondeterministic automata is essentially bisimilarity, and thus
  does not coincide with language equivalence. Language equivalence
  can be captured as behavioural equivalence on the determinization,
  which is obtained via the standard powerset construction. This
  construction can be lifted to coalgebraic generality, assuming a
  so-called Eilenberg-Moore distributive law between the functor
  determining the type of accepted structure (e.g.\ word languages)
  and a monad capturing the branching type (e.g.\ nondeterministic,
  weighted, probabilistic). Eilenberg-Moore-style coalgebraic
  semantics in this sense has been shown to be essentially subsumed by
  the more general framework of graded semantics, which is centrally
  based on graded monads. Graded semantics comes with a range of
  generic results, in particular regarding invariance and, under
  suitable conditions, expressiveness of dedicated modal logics for a
  given semantics; notably, these logics are evaluated on the original
  state space. We show that the instantiation of such graded logics to
  the case of Eilenberg-Moore-style semantics works extremely
  smoothly, and yields expressive modal logics in essentially all
  cases of interest. We additionally parametrize the framework over a
  quantale of truth values, thus in particular covering both the
  two-valued notions of equivalence and quantitative ones, i.e.\
  behavioural distances.
\end{abstract}

\section{Introduction}
\label{sec:introduction}

When dealing with the logical aspects of state-based systems, one is
particularly interested in the property of \emph{expressiveness}, that
is, the ability of a logic to differentiate between states that behave
in different ways. The prototypical example of this property is
captured by the \emph{Hennessy-Milner theorem}
\cite{DBLP:conf/icalp/HennessyM80}, with modal logic distinguishing
states in finitely branching transition systems precisely up to
bisimilarity. There is, however, a wide array of equivalences of
interest that are coarser than bisimilarity
\cite{g:linear-branching-time}, each necessitating a different type of
logic to stay invariant under the semantics while ensuring
expressiveness. A similar story unfolds when state-based systems are
studied abstractly as coalgebras for a given functor that encapsulates
the transition type of systems~\cite{DBLP:journals/tcs/Rutten00}: The
finest and mathematically most convenient type of equivalence is given
by coalgebraic behavioural equivalence, with much of the literature on
coalgebraic logic focusing on expressiveness with respect to this type
of equivalence (e.g.\
\cite{DBLP:journals/ndjfl/Pattinson04,DBLP:journals/tcs/Schroder08,DBLP:journals/jcss/MartiV15,DBLP:conf/concur/KonigM18,DBLP:journals/lmcs/WildS22,DBLP:conf/csl/Forster0HNSW23}),
though this might not necessarily be the equivalence the application
demands. Consider for example nondeterministic automata,
i.e. coalgebras for the $\Set$-endofunctor
$2 \times (\mathcal{P}{-})^\Sigma$. The equivalence of interest in
these systems is language equivalence, and as such is potentially much
coarser than the coalgebraic notion of behavioural equivalence, which
in this case instantiates to bisimilarity. A possible way to deal with
this mismatch is to first transform the nondeterministic automaton
into a deterministic one, that is, a coalgebra for the
$\Set$-endofunctor $F=2 \times (-)^\Sigma$, via the powerset
construction, obtaining language equivalence as behavioural
equivalence in the determinized automaton. The powerset construction
generalizes to coalgebras for functors of the form $FT$ where~$F$ is a
functor capturing the type of accepted structure (e.g.\ word languages
for~$F=2 \times (-)^\Sigma$ as above) and a monad~$T$ capturing the
branching type of systems ($T=\power$ as above captures
nondeterminism; other choices of~$T$ capture, e.g., probabilistic or
weighted branching). To be applicable, this approach requires a
so-called Eilenberg-Moore distributive law of~$T$ over~$F$
\cite{DBLP:conf/fsttcs/SilvaBBR10}; it then equips $FT$-coalgebras
with a language-type semantics determined by~$F$, to which we refer as
\emph{Eilenberg-Moore semantics}. 

Our present aim is to obtain modal logics that are expressive and
invariant for Eilenberg-Moore semantics, and at the same time can be
seen as fragments of the standard expressive branching-time
coalgebraic modal logics (in analogy to logics for the
linear-time/branching-time spectrum of labelled transition system
semantics~\cite{g:linear-branching-time}, which are fragments of
standard Hennessy-Milner logic). To this end, we exploit the machinery
of \emph{graded semantics}
\cite{DBLP:conf/calco/MiliusPS15,DBLP:conf/concur/DorschMS19},
in which notions of behavioural equivalence are modelled by mapping
into a \emph{graded monad}~\cite{Smirnov08}; it has been shown that
Eilenberg-Moore semantics can essentially be cast as a graded
semantics~\cite{DBLP:conf/birthday/KurzMPS15,DBLP:conf/calco/MiliusPS15}. 
Graded semantics comes with a general notion of invariant \emph{graded
  logic} and a criterion for a graded logic to be
expressive~\cite{DBLP:conf/calco/MiliusPS15,DBLP:conf/concur/DorschMS19}.

\paragraph*{Contribution} By instantiating the expressivity criterion
of the graded semantics framework to Eilenberg-Moore semantics, we
show that it is sufficient to provide a set of modal operators that
separate the elements of $FX$, while the treatment of~$T$ is
automatically provided by the framework. Separation of~$FX$ is
typically easy to ensure, justifying the slogan that Eilenberg-Moore
semantics essentially always admits an expressive invariant logic. We
parametrize our results over the choice of a quantale that serves as a
domain of both distances and truth values, allowing an instantiation to both
the two-valued setting, where states are either equivalent or not and
formulae take binary truth values, and to quantitative settings, where
similarity of states is a continuum and formulae may take intermediate
values, for instance in the real unit interval. We thus in particular
cover notions of \emph{behavioural distance}
(e.g.~\cite{bw:behavioural-pseudometric,DBLP:conf/fsttcs/BaldanBKK14,bbkk:coalgebraic-behavioral-metrics}),
providing logics that are expressive in the sense that the behavioural
distance between states is always witnessed by differences in the
evaluation of suitable formulae. We discuss a range of examples, in
some cases obtaining new characteristic modal logics, e.g.\ for
probabilistic trace equivalence of reactive probabilistic automata
with black-hole termination.

\paragraph*{Related work} There has been a fair amount of work on
the coalgebraic treatment of system semantics beyond branching time.
Approaches using Kleisli-type distributive laws~\cite{HasuoEA07} and
Eilenberg-Moore distributive laws
(e.g.~\cite{DBLP:conf/fsttcs/SilvaBBR10,KissigKurz10,BonchiEA12,DBLP:journals/jcss/Jacobs0S15})
are subsumed by graded
semantics~\cite{DBLP:conf/birthday/KurzMPS15,DBLP:conf/calco/MiliusPS15}. The
Kleisli approach has also been applied to infinite-trace semantics
(e.g.~\cite{DBLP:journals/entcs/Jacobs04a,DBLP:journals/corr/KerstanK13,DBLP:conf/calco/UrabeH15,DBLP:journals/fuin/Cirstea17a})
 and to trace semantics via steps~\cite{DBLP:journals/logcom/RotJL21}. Klin and
Rot~\cite{KlinRot15} present a notion of semantics based on selecting
a modal logic, which is then expressive by definition of the
semantics. For our present purposes, the most closely related piece of
previous work uses corecursive algebras as a unifying concept
subsuming the Kleisli-based, Eilenberg-Moore-based, and logic-based
approaches to coalgebraic trace
semantics~\cite{DBLP:journals/logcom/RotJL21}. In particular, the
comparison between the Eilenberg-Moore-based and the logic-based
semantics in this
framework~\cite[Section~7.1]{DBLP:journals/logcom/RotJL21} can be read
as an expressiveness criterion for logics over Eilenberg-Moore
semantics. In relation to this criterion, the distinguishing feature
of our present main result lies in the concreteness of the
construction of the logic in terms of modal and propositional
operators, as well as the ease of checking our expressiveness
criterion, which comes essentially for free in all cases of
interest. We note also that our criterion works in quantalic
generality, and thus applies also to notions of behavioural distance,
which so far are not covered in the approach via corecursive algebras.

Through its applicability to behavioural distances, our work relates
additionally to a spate of recent work on the coalgebraic treatment of
characteristic logics for behavioural distances. For the most part,
such results have been concerned mainly with branching-time distances
(e.g.~\cite{DBLP:conf/concur/KonigM18,DBLP:journals/lmcs/WildS22,kkkrh:expressivity-quantitative-modal-logics,DBLP:conf/csl/Forster0HNSW23}).

Kupke and Rot~\cite{DBLP:journals/lmcs/KupkeR21} study logics for
\emph{coinductive predicates}, generalizing branching-time behavioural
distances. Our overall setup differs from the one used
in~\cite{DBLP:journals/lmcs/KupkeR21} by working with coalgebras for
functors that live natively on metric spaces, including such functors
that are not liftings of a set functor.

In recent work by König and (some of) the present
authors~\cite{bgkm:hennessy-milner-galois,DBLP:conf/stacs/bgk24},
expressive logics for coalgebraic trace-type behavioural distances
have been approached by setting up Galois connections between logics
and distances. This concept is highly general (and in fact not even
tied to models being coalgebras) but requires a comparatively high
amount of effort for concrete instantiations. Moreover, its focus is
on fixpoint characterizations of logical distance rather than on
expressiveness w.r.t.\ a given notion of behavioural distance, and in
fact the behaviour function inducing behavioural distance is defined
directly via the logic.

\section{Preliminaries}
\label{sec:preliminaries}

We assume basic familiarity with category theory
(e.g.~\cite{AdamekHerrlich90}). In the following, we recall requisite
definitions and facts on universal coalgebra, quantales, and lifting
functors to categories of monad algebras.

\subsection{Universal Coalgebra}

State-based systems of various types, such as non-deterministic,
probabilistic, weighted, or game-based transition systems, are treated
uniformly in the framework of \emph{universal
  coalgebra}~\cite{DBLP:journals/tcs/Rutten00}. The branching type of
a system is encapsulated as a functor $G\colon\Ccat\to\Ccat$ on a
suitable base category~$\Ccat$, for instance on the category $\Set$ of
sets and maps. A \emph{$G$-coalgebra} $(C,c)$ then consists of a
$\Ccat$-object~$C$, thought of as an object of \emph{states}, and a
morphism $c\colon C\to GC$, thought of as a \emph{transition map} that
assigns to each state a structured collection of successor states,
with structure determined by~$G$. For instance, on $\Ccat=\Set$, a
$\power$-coalgebra for the covariant powerset functor is just a
nondeterministic transition system, while a $G$-coalgebra for the
functor~$G$ given by $GX=2\times X^\Sigma$, with~$\Sigma$ a fixed
\emph{alphabet}, is a deterministic automaton (without initial state),
assigning to each state a finality status and a tuple of successors,
one for every letter in~$\Sigma$.

A \emph{morphism} $h\colon (C,c)\to(D,d)$ of $G$-coalgebras is a
$\Ccat$-morphism $h\colon C\to D$ that is compatible with the
transition maps in the sense that $d\cdot h=Gh\cdot c$. States
$x,y\in C$ in a coalgebra $(C,c)$ are \emph{behaviourally
  equivalent} if there exist a $G$-coalgebra $(D,d)$ and a morphism
$h\colon (C,c)\to (D,d)$ such that $h(x)=h(y)$. For instance, two
states in a labelled transition system (i.e.\ a coalgebra for
$G=\power(\Sigma\times (-))$ where~$\Sigma$ is the set of labels) are
behaviourally equivalent iff they are bisimilar in the usual sense.

The (initial $\omega$-segment of) the \emph{final chain} of~$G$ is the
sequence $(G^n1)_{n<\omega}$ of $\Ccat$-objects. Given a $G$-coalgebra
$(C,c)$, we have the \emph{canonical cone} of maps
$c_n\colon C\to G^n 1$, defined by~$c_0$ being the unique map $C\to 1$
and by ${c_{n+1}=C\xrightarrow{c}
GC\xrightarrow{Gc_n}G^{n+1}1}$. When $\Ccat$ is a concrete category over $\Set$,
states $x,y\in C$ are termed
\emph{finite-depth behaviourally equivalent} if $c_n(x)=c_n(y)$ for
all $n\in\nat$. For finitary set functors, finite-depth behavioural
equivalence and behavioural equivalence coincide~\cite{Worrell00}.

\subsection{Quantales}\label{sec:quantales}

We use (symmetrized) \emph{quantale-enriched categories} as a joint
generalization of equivalence relations and pseudometric spaces; this
enables us to cover both two-valued and quantitative semantics and
logics uniformly in one framework. In a nutshell, a quantale is a
monoid in the category of complete join semilattices. Explicitly, this
notion expands as follows:

\begin{defn} A (commutative unital) \emph{quantale}
  $\V = (V, \otimes, k, \leq)$ consists of a set $V$ that carries both
  the structure of a complete lattice $(V, \leq)$ and the structure of
  a commutative monoid $(V, \otimes, k)$ such that for all $v \in V$,
  the operation ${-} \otimes v$ is join-continuous; that is,
  \begin{equation*}\textstyle
    \left(\bigvee_{i\in I} u_i \right)\otimes v = \bigvee_{i\in I} \left(u_i \otimes v\right)
  \end{equation*}
  where we use $\bigvee$ to denote joins.
\end{defn}

\noindent By the standard equivalence between join preservation and
adjointness for functions on complete lattices, it follows that for
every $b\in V$, the map ${-} \otimes b $ has a right adjoint
$[b, {-}]$, with defining property
\begin{equation*}
  a \otimes b \leq c  \Leftrightarrow a \leq [b,c]
\end{equation*}
As first observed by Lawvere~\cite{lawvere1973metric}, metric spaces
can be seen as enriched categories, which leads to the notion of
categories enriched in a quantale~$\V$, or briefly $\V$-categories, as
a generalized notion of (pseudo-)metric space:
\begin{defn}
  A \emph{$\V$-category} is a pair $(X, d_X)$ consisting of a set $X$ and a
  function $d_X\colon X\times X\to V$ such that for all
  $x, y, z \in X$ we have $d_X(x, y) \otimes d_X(y, z) \leq d_X(x, z)$,
  as well as $k \leq d_X(x, x)$. A $\V$-category $(X, d_X)$ is
  \emph{symmetric} if $d_X(x, y) = d_X(y, x)$ for all $x,y\in X$, and
  \emph{separated} if $k \le d_X(x,y)$ implies $x=y$. A function
  $f\colon X \to Y$ is a \emph{$\V$-functor} between $\V$-categories
  $(X, d_X)$ and $(Y, d_Y)$ if $d_X(a, b) \leq d_Y(f(a), f(b))$ for
  all $a, b \in X$.
\end{defn}

\noindent We fix a quantale~$\V$ for the rest of the technical
development. We write $\dpmet_\V$ for the category of $\V$-categories
and $\V$-functors, which we view as generalized directed pseudometric
spaces, with distance values in~$\V$. Further, we write $\pmet_\V$ for
the full subcategory of symmetric $\V$-categories, viewed as
generalized pseudometric spaces, and $\met_\V$ for the full
subcategory of symmetric and separated $\V$-categories, viewed as
generalized metric spaces.  The quantale $\V$ itself
has the structure of an object in $\dpmet_\V$, where
$d(x, y) = [x, y]$ for all $x, y \in \V$. It may also be viewed as an
object in $\met_\V$ through symmetrization:
$d_\text{sym}(x, y) = [x, y] \wedge [y,x]$. In this way, we will often
use $\V$ as the codomain of evaluation morphisms of our logics.
We will focus on the following two examples:
\begin{expl}\label{expl:quantales}
  \begin{enumerate}[wide]
  \item\label{item:two} The lattice $2=\{\bot,\top\}$ carries a
    quantale $\mathbf{2} = (2, \wedge, \top, {\leq})$. In this case,
    $[b,c]$ is just the Boolean implication $b\to c$. The category
    $\pmet_\mathbf{2}$ is isomorphic to the category of setoids, i.e.\
    of equivalence relations and equivalence-preserving maps, while
    the category $\met_\textbf{2}$ is isomorphic to the category of
    sets and functions. We use this quantale to cover two-valued
    equivalences, used in situations where one is only interested in
    determining whether states behave in precisely the same way or
    not.
  \item We use the quantale $[0,1]_\oplus = ([0,1], \oplus, 0, \geq)$,
    where~$\oplus$ is truncated addition ($a\oplus b=\min(a+b,1)$), to
    cover cases where one wishes to measure differences in the
    behaviour of states in a continuous manner. In this case, $[-,-]$
    is truncated subtraction ($[b,c]=\max(c-b,0)$). Indeed, taking
    $[-,1]$ as negation makes $[0,1]_\oplus$ into an MV-algebra,
    providing a domain of truth values for multi-valued \L{}ukasiewicz
    logic.  The category $\met_{[0,1]_\oplus}$ is isomorphic to the
    usual category of $1$-bounded metric spaces and non-expansive
    maps~\cite{lawvere1973metric}, while $\pmet_{[0,1]_\oplus}$ is isomorphic to the category
    of pseudometric spaces (that is, distinct elements may take
    distance $0$). Note that the ordering on the set $[0,1]$ is
    reversed compared to its natural ordering. This is necessary,
    since otherwise $\oplus$ does not   distribute over the
    empty join.
	\end{enumerate}
\end{expl}

\noindent We will use the concept of initiality (in the concrete case of $\V$-categories)
to describe the fact that a set of morphisms is large enough to witness the
distances in its domain. Later, expressivity demands that the set of
evaluation morphisms of formulae form an initial source.
\begin{defn}

  A source $\mathfrak{A}$ of $\V$-functors
  $f_i\colon (X, d_X) \to (Y_i, d_{Y_i})$ is \emph{initial} if
  $d_X(x, y) = \bigwedge_{i \in I}d_{Y_i}(f_i(x), f_i(y))$ for all
  $x, y \in X$.

\end{defn}

\subsection{Lifting Functors to Eilenberg-Moore Categories}\label{sec:monads}

Recall that a \emph{monad} $(T,\mu,\eta)$, denoted just~$T$ by abuse
of notation, on a base category~$\Ccat$ consists of a functor
$T\colon\Ccat\to\Ccat$ and natural transformations
${\mu\colon TT\Rightarrow T}$, as well as $\eta\colon\Id\Rightarrow T$ (the \emph{multiplication}
and \emph{unit} of~$T$) satisfying natural laws. Monads on~$\Set$
may be thought of as encapsulating algebraic theories, with~$TX$ being
terms over~$X$ modulo provable equality, $\mu$ collapsing layered
terms into terms, thus abstracting substitution, and~$\eta$ converting
variables into terms. We call a monad~$T$ \emph{affine} \cite{DBLP:conf/cmcs/Jacobs16}
when $T$ preserves the terminal
object, that is $T1 \cong 1$. For
instance, the \emph{distribution monad} $\dist$, given by $\dist X$
being the set
\begin{equation*}
  \{f \colon X\to[0,1]\mid f(x)=0\text{ for almost all $x\in X$, }
  \textstyle\sum_{x\in X}f(x)=1\}
\end{equation*}
of finitely supported probability
distributions on~$X$, is affine. Monads induce a natural notion of
algebra: A \emph{monad algebra} or \emph{Eilenberg-Moore algebra}
$(A,a)$ for~$T$ consists of a $\Ccat$-object~$A$ and a morphism
$a\colon TA \to A$ making the left and middle diagrams below commute.
\begin{equation*}
  \begin{tikzcd}
    A \arrow[r, "\eta_A"]\arrow[dr, "\id_A" below left] & TA \arrow[d, "a"]\\
    & A
  \end{tikzcd}
  \qquad
  \begin{tikzcd}
  TTA \arrow[r, "Ta"] \arrow[d, "{\mu_{A}}" left]  & TA \arrow[d, "{a}"] \\
  TA \arrow[r, "{a}" below]                        & A
\end{tikzcd}
\qquad
\begin{tikzcd}
  TA \arrow[r, "{Tf}"] \arrow[d, "a" left]  & TB \arrow[d, "{b}"] \\
  A \arrow[r, "f" below]                        & B
\end{tikzcd}
\end{equation*}
A $\Ccat$-morphism $f\colon A \to B$ is a morphism between algebras
$f\colon (A, a) \to (B,b)$ if the right diagram commutes.  We write
$\alg{T}$ for the category of Eilenberg-Moore algebras for~$T$ and
their morphisms. We denote the functor that takes a $\Ccat$-object~$A$
to the free $T$-algebra $(TA, \mu)$ over~$A$ by
$L\colon \Ccat \to \alg{T}$. This functor is left adjoint to the
forgetful functor $R\colon \alg{T} \to \Ccat$ that takes algebras
$(A, a)$ to their carrier $A$.  The category $\alg{T}$ has all limits
that $\Ccat$ has \cite[Proposition 20.12]{AdamekHerrlich90}. We
occasionally need the $n$-fold power $(A, a)^n$ of an algebra
$(A, a)$, whose carrier is the $\Ccat$-object $A^n$. We denote its
algebra structure by $a^{(n)}\colon T(A^n) \to A^n$.

Coalgebraic determinization~\cite{DBLP:conf/fsttcs/SilvaBBR10} is
concerned with coalgebras for functors of the form $G=FT$ where~$T$ is
a monad, thought of as capturing the branching type of systems,
and~$F$ is a functor determining the system semantics. As indicated in
the introduction, the basic example is given by nondeterministic
automata over an alphabet~$\Sigma$, which are coalgebras for the set
functor $G=FT$ with $FX=2\times X^\Sigma$ and $T=\power$, while
$F$-coalgebras are deterministic automata. The coalgebraic
generalization of the powerset construction that determinizes
nondeterministic automata relies on having a suitable type of
distributive law between~$F$ and~$T$:

\begin{defn}
  \label{def:em-laws}
  An \emph{Eilenberg-Moore distributive law}, or just \emph{EM law}, of a monad
  $(T, \mu, \eta)$ over a functor $F$ is a natural transformation
  $\zeta\colon TF\Rightarrow FT$ such that the following diagrams
  commute:
  \begin{equation*}
    \begin{tikzcd}
      F \arrow[rd, "F\eta"] \arrow[r, "\eta F"] & TF \arrow[d, "\zeta"] \\
      & FT
    \end{tikzcd}
    \qquad
	  \begin{tikzcd}
      TTF \arrow[d, "\mu F"] \arrow[r, "T\zeta"] & TFT \arrow[r, "\zeta T"] & FTT \arrow[d, "F\mu"] \\
      TF \arrow[rr, "\zeta"]                     &                          & FT
    \end{tikzcd}
  \end{equation*}
\end{defn}
\noindent It is well-known (cf.\ \cite{MULRY2002:LiftingAlgebras}) that EM laws
$\zeta \colon TF \Rightarrow FT$ are in 1-1 correspondence with
liftings $\tilde{F}$ of the functor $F$ to the Eilenberg-Moore
category $\mathbf{EM}(T)$. Given~$\zeta$, the functor $\tilde F$ maps the
$T$-algebra $(A, a)$ to $(FA, Fa \cdot \zeta_A)$. As a result, every
$FT$-coalgebra $c\colon X\to FTX$ can be determinized in the presence
of an EM law~\cite{DBLP:conf/fsttcs/SilvaBBR10}, yielding an
$\tilde F$-coalgebra $c^\#\colon TX\to FTX$ in $\alg{T}$ as follows:
\[
TX \xrightarrow{Tc} TFTX \xrightarrow{\zeta_{TX}} FTTX \xrightarrow{F\mu_X} FTX
\]
Taking a more abstract perspective, $c^\#$ is the adjoint transpose of~$c$ under
${L\dashv R}$. We say that states $c,d\in C$ are \emph{EM-equivalent} if
$\eta_C(c),\eta_C(d)\in TC$ are behaviourally equivalent in the
$\tilde{F}$-coalgebra $(TC,c^\#)$. We refer to this equivalence as
\emph{EM semantics}; when this equivalence can be captured as the
kernel of a suitable map (in this case, the map assigning to each
state its accepted language), we also refer to this map as the EM
semantics. We will later encounter situations where the codomain of
the semantics carries a generalized metric structure, in which case we
will also subsume the induced generalized pseudometric on~$C$ under
the moniker `EM semantics'.

The standard powerset construction for determinizing nondeterministic automata is recovered by the following EM law:
\begin{expl}\label{ex:powerset-construction}
	In $\textbf{Set}$ (i.e. $\met_\textbf{2}$), let $T = \mathcal{P}$ be the powerset monad and $F = 2 \times {-}^\Sigma$. 
  The determinization $(\mathcal{P}C, c^{\#})$ of an $FT$-coalgebra
  $(C,c)$ w.r.t.\ the EM law $\zeta\colon TF\Rightarrow FT$ defined by
  \begin{equation*}\textstyle
    \zeta(t) = \Big( \bigvee_{(v, f)\in t} v,\quad \lambda a.\{ f(a)
    \mid (v, f) \in t\}\Big)
  \end{equation*}
  for $t \in \power(2 \times X^\Sigma)$ is precisely the powerset
  construction. Thus, the language semantics of nondeterministic
  automata is an instance of EM semantics.
\end{expl}

\begin{expl}
  \label{ex:distributive-machine}
  More generally, let $F = A \times {-}^\Sigma$ where $\Sigma$ is discrete, let~$T$ be a monad on
  $\Set$, and suppose that~$A$ carries a $T$-algebra structure
  $a \colon TA \to A$.  Define a natural transformation
  $\delta\colon T({-}^\Sigma) \Rightarrow (T{-})^\Sigma$ by
  $\delta(t)(\sigma) = T(\lambda f. f(\sigma))(t)$.  We then have an
  EM law
  $\zeta \colon T(A \times {-}^\Sigma) \Rightarrow A\times
  {(T-)}^\Sigma$ given (componentwise) by
  \begin{align*}
   \pi_1 \cdot \zeta_X & = (T(A \times {X}^\Sigma) \xrightarrow{T\pi_1} TA \xrightarrow{a} A)\\
     \pi_2 \cdot \zeta_X &= (T(A \times {X}^\Sigma) \xrightarrow{T\pi_2} T({X}^\Sigma) \xrightarrow{\delta} (TX)^\Sigma).
  \end{align*}
  The arising EM semantics assigns to each state~$x$ a map
  $\Sigma^*\to A$, which may be thought of as assigning to each word
  $w\in\Sigma^*$ the degree (a value in~$A$) to which~$x$ accepts~$w$.
\end{expl}

\section{Graded Semantics and Graded Logics}

\emph{Graded
  semantics}~\cite{DBLP:conf/calco/MiliusPS15,DBLP:conf/concur/DorschMS19}
uniformly captures a wide range of semantics on various system types
and of varying degrees of granularity as found, for instance, on the
linear-time/branching-time spectrum of labelled transition system
semantics~\cite{g:linear-branching-time}. Here, we are interested
primarily in applying general results provided by the framework of
graded semantics to the setting of EM semantics, which is, in essence,
subsumed by graded
semantics~\cite{DBLP:conf/birthday/KurzMPS15,DBLP:conf/calco/MiliusPS15}. We
recall the basic definition of graded semantics as such, and then give
a new perspective on a general notion of characteristic modal logics
for graded semantics, so-called graded logics.

\subsection{Graded Semantics}

The concepts central to graded semantics are those of graded monads
and graded algebras. These are very similar to those of monads and
monad algebras as recalled in Section~\ref{sec:monads} but, in the
mentioned analogy with universal algebra, equip operations and terms
with a \emph{depth} that, in the application to system semantics,
records the depth of look-ahead; that is, the depth corresponds to the
(exact) number of transition steps considered. We briefly review the
formal definitions.

\begin{defn}[Graded Monad]
    A \emph{graded monad} $\mathbb{M}$ on a category\;$\Ccat$ consists of a family of functors $M_n\colon \Ccat \to \Ccat$ for $n \in \nat$, a natural transformation ${\eta\colon  Id \Rightarrow M_0}$ (the \emph{unit}), and a family of natural transformations ${\mu^{n,k}\colon M_nM_k \Rightarrow M_{n+k}}$ for all $n, k \in \nat$ (the \emph{multiplication}) such that for all $n, k, m \in \nat$ the following diagrams commute:

\begin{equation*}
\begin{tikzcd}
& M_n \arrow[rd, "M_n \eta"] \arrow[ld, "\eta M_n"'] \arrow[d, "id_{M_n}"] &\\
M_0M_n \arrow[r, "\mu^{0,n}"] & M_n & M_nM_0 \arrow[l, "\mu^{n,0}"']
\end{tikzcd}
\qquad
\begin{tikzcd}[column sep=huge]
M_nM_kM_m \arrow[r, "M_n\mu^{k,m}"] \arrow[d, "\mu^{n,k}M_m"]  & M_nM_{k+m} \arrow[d, "{\mu^{n,k+m}}"] \\
M_{n+k}M_m \arrow[r, "{\mu^{n+k,m}}"]                        & M_{n+k+m}
\end{tikzcd}
\end{equation*}
\end{defn}

\begin{defn}[Graded semantics]\label{def:graded-sem}
  A \emph{graded semantics} $(\alpha, \mathbb{M})$ for an endofunctor
  $G\colon \Ccat \to \Ccat$ consists of a graded monad $\mathbb{M}$ on
  $\Ccat$ and a natural transformation $\alpha\colon G \Rightarrow M_1$. If
  $(C, c)$ is a $G$-coalgebra, then we define the $n$-step behaviour
  $c^{(n)}\colon C \to M_n1$, for $n \in \nat$, by
  \begin{equation*}
    c^{(0)}=( X \xrightarrow{M_0! \cdot \eta} M_01)
    \qquad c^{(n+1)}( X \xrightarrow{\alpha \cdot c} M_1X \xrightarrow{M_1c^{(n)}} M_1M_n1 \xrightarrow{\mu^{1n}} M_{n+1}1).
  \end{equation*}
\end{defn}

\noindent We think of $c^{(n)}$ as assigning to a state in~$C$ its behaviour
after $n$ steps. We illustrate this more concretely in
Example~\ref{ex:gradedMonads}. 
We are mainly interested in the case where the base category~$\Ccat$
is a category of generalized (directed) pseudometric spaces
(Section~\ref{sec:quantales}). In this case, a graded semantics
induces a notion of behavioural distance:

\begin{defn}[Behavioural distance]
  \label{def:graded-trace-distance}
  When $\Ccat$ is $\dpmet_\V$ (or $\pmet_\V$, $\met_\V$), then we
  define the \emph{behavioural distance} of two states $x, y \in C$ of
  a $G$-coalgebra $(C,c)$ under a graded semantics
  $(\alpha, \mathbb{M})$ to be
	\begin{equation*}\textstyle
		d^{b}(x, y) = \bigwedge_{n \in \nat}d_{M_n1}(c^{(n)}(x),
  c^{(n)}(y))	.
	\end{equation*}
\end{defn}
\begin{rem}
	In case $\V=\mathbf{2}$ (Example~\ref{expl:quantales}.\ref{item:two}),
  behavioural distance is two-valued, and thus in fact constitutes
  either a preorder (if $\Ccat$ is $\dpmet_\V$) or an equivalence (if
  $\Ccat$ is $\pmet_\V$).
\end{rem}

\begin{expl}
  \label{ex:gradedMonads}
  We recall two basic examples of graded
  monads~\cite{DBLP:conf/calco/MiliusPS15} and associated
  graded semantics, capturing branching-time semantics and EM
  semantics, respectively. In both cases, it happens that~$\alpha$ is
  identity; this need not always be the case,
  however~\cite{DBLP:conf/concur/DorschMS19}.
  \begin{enumerate}[wide]
  \item\label{item:branching} Any functor $F$ induces a graded monad
    $\mathbb{M}_F$ where the functor parts $M_n = F^n$ are given via
    repeated application of $F$ and both multiplication and unit are
    identity. The arising graded semantics of $F$-coalgebras is
    branching-time semantics, specifically finite-depth behavioural
    equivalence (which coincides with behavioural equivalence if~$F$
    is finitary).

  \item\label{item:emgm} Any EM law $\zeta\colon TF \Rightarrow FT$
      induces a graded monad $\mathbb{M}_\zeta$ where ${M_n = F^n T}$. The
    unit of $\mathbb{M}_\zeta$ is the unit of $T$.  We define an
    iterated distributive law
          ${\zeta^{(n)}\colon TF^n \Rightarrow F^n T}$ by putting
    \begin{equation*}
      \zeta^{(0)} = \id\quad\text{and}\quad
      \zeta^{(n+1)} = TF^{n+1} \xrightarrow{\zeta F^n} FTF^n
      \xrightarrow{F\zeta^{(n)}} F^{(n+1)}T.
  \end{equation*}
  The multiplications of the graded monad $\mathbb{M}_\zeta$ are then
          given by~${\mu^{m,n} =F^{n+m}\mu \cdot F^m\zeta^{(n)}T}$. The arising
  graded semantics is essentially EM semantics, in the sense that the
  latter is obtained by erasing further information by postcomposing
  the maps $c^{(n)}\colon C\to F^n T1$ (in the notation of
  Definition~\ref{def:graded-sem}) with $F^n!$ where $!$ is the unique
  map $T1\to 1$. In particular, the EM semantics and the graded
  semantics introduced by an EM law for~$T$ agree exactly if~$T$ is
  affine (Section~\ref{sec:monads}). Otherwise, the information erased
  by $F^n!$ essentially concerns the possibility of executing certain
  words, without regard to their
  acceptance~\cite[Section~5]{DBLP:conf/birthday/KurzMPS15}. For a
  concrete example where~$T$ is affine, consider $T=\dist$ (the
  distribution monad, cf.\ Section~\ref{sec:monads}) and
  $FX=[0,1]\times X^\Sigma$, with an EM law~$\zeta$ as per
  Example~\ref{ex:distributive-machine}. Then $FT$-coalgebras are
  reactive probabilistic automata, and for a state~$x$ in an
  $FT$-coalgebra,
  $c^{(n)}(x)\in F^n\dist 1\cong F^n1\cong [0,1]^{\Sigma^{<n}}$ assigns
  to each word of length~$<n$ over~$\Sigma$ its probability of being
  accepted.
  \end{enumerate}
  Note that~\ref{item:branching}.\ is the special case
  of~\ref{item:emgm}.\ where $T=\Id$.
\end{expl}
\noindent We will in fact be interested exclusively in graded monads
that are, in the universal-algebraic
view~\cite{DBLP:conf/calco/MiliusPS15,DBLP:conf/concur/DorschMS19},
presented by operations and equations of depth at most~$1$, which
intuitively means that identifications among behaviours do not depend
on looking more than one step ahead. Categorically, this notion is
captured as
follows~\cite[Proposition~7.3]{DBLP:conf/calco/MiliusPS15}:
\begin{defn}
  \label{def:depthone}
  We say that a graded monad is \emph{depth-$1$} if for all
  $n\in \mathbb{N}$, $\mu^{1,n}$ is a coequalizer in the following
  diagram:
\begin{equation*}
\begin{tikzcd}
	M_1M_0M_nX \arrow[r,swap, "\mu^{1,0}M_n", shift right] \arrow[r, "M_1\mu^{0,n}", shift left] & M_1M_nX \arrow[r, "\mu^{1,n}"] & M_{1+n}X.
\end{tikzcd}
\end{equation*}
\end{defn}

\begin{expl}
  All graded monads described in Example~\ref{ex:gradedMonads} are depth-1.
\end{expl}

\noindent The semantics of modalities in graded logics will rely on a
graded variant of the notion of monad algebra:

\begin{defn}[Graded algebra]
  Let $\mathbb{M}$ be a graded monad in $\Ccat$, and
    ${n\in\nat \cup \{\omega\}}$. A \emph{graded $M_n$-algebra}
  $A = ((A_k)_{k \leq n}, (a^{m,k})_{m+k \leq n})$ consists of a family
  of $\Ccat$-objects $A_k$ and morphisms
  $a^{m,k}\colon M_m A_k \to A_{m+k}$ satisfying the following
  conditions: For $m \leq n$, we have
  $a^{0,m}\cdot\eta_{A_m} = \id_{A_m}$ and additionally, if
  $m + r + k \leq n$, then the left diagram below commutes:
  \begin{equation*}
    \begin{tikzcd}[column sep=large]
      M_mM_rA_k \arrow[r, "M_ma^{r,k}"] \arrow[d, "\mu^{m,r}_{A_k}"]  & M_mA_{r+k} \arrow[d, "{a^{m,r+k}}"] \\
      M_{m+r}A_k \arrow[r, "{a^{m+r,k}}"]                        & A_{m+r+k}
    \end{tikzcd}
    \qquad
	  \begin{tikzcd}[column sep=large]
      M_mA_k \arrow[r, "{M_mf_k}"] \arrow[d, "a^{m,k}"]  & M_mB_k \arrow[d, "{b^{m,k}}"] \\
      A_{m+k} \arrow[r, "{f_{m+k}}"]                        & B_{m+k}
    \end{tikzcd}
  \end{equation*}
  A \emph{homomorphism} of $M_n$-algebras $A$ and $B$ is a family of
  maps $f_k\colon A_k \to B_k$ such that the above right diagram
  commutes for all $m + k \leq n$. For all
  $n \in \mathbb{N} \cup \{\omega\}$, the collection of $M_n$-algebras
  and their morphisms forms a category $\Galg{n}{\mathbb{M}}$.
\end{defn}

\noindent The category $\Galg{0}{\mathbb{M}}$ is the Eilenberg-Moore
category $\alg{M_0}$ for the (non-graded) monad
$(M_0, \eta, \mu^{0,0})$. The semantics of modalities in graded logics
will involve a special type of
$M_1$-algebras~\cite{DBLP:conf/concur/DorschMS19}:

\begin{defn}[Canonical algebras]
  \label{def:canonical}
  For $i \in \{0,1\}$, let
  $({-})_i: \Galg{1}{\mathbb{M}}\rightarrow \Galg{0}{\mathbb{M}}$ be
  the functor taking an $M_1$-algebra
  $A=((A_k)_{k \leq 1}, (a^{m,k})_{m+k \leq 1})$ to the $M_0$-algebra
  $(A_i, a^{0,i})$. We say that an $M_1$-algebra $A$ is
  \emph{canonical} if it is free over $({-})_0$, i.e.\ if for all
  $M_1$-algebras $B$ and $M_0$-homomorphisms
  $f: (A)_0 \rightarrow (B)_0$ there is a unique $M_1$-homomorphism
  $g: A \rightarrow B$ such that $(g)_0 = f$.
\end{defn}

\begin{lem}\label{lem:canonical}(\cite[Lemma 5.3]{DBLP:conf/concur/DorschMS19})
  An $M_1$-algebra~$A$ is canonical iff the following diagram
 is a coequalizer diagram in the category of $M_0$-algebras:
	\begin{equation*}\label{eq:coeq}
	\begin{tikzcd}
		M_1M_0A_0 \arrow[r,swap, "\mu^{1,0}", shift right] \arrow[r, "M_1a^{0,0}", shift left] & M_1A_0 \arrow[r, "a^{1,0}"] & A_1
	\end{tikzcd}
	\end{equation*}
\end{lem}

\noindent Combining Definition~\ref{def:depthone} with Lemma~\ref{lem:canonical}
immediately gives us the following
fact~\cite{DBLP:conf/concur/DorschMS19}, which is a crucial ingredient
for invariance of graded logics:
\begin{prop}
  \label{canonical}
  If\/ $\mathbb{M}$ is a depth-1 graded monad, then for every
  $n\in \nat$ and every object~$X$, the $M_1$-algebra with carriers
  $M_nX$, $M_{n+1}X$ and multiplications as algebra structure is
  canonical.
\end{prop}
\subsection{Graded Logics as a Fragment of Branching-Time Logic}
\label{subsec:glogic-btime}
We proceed to recall the general framework of (branching time)
coalgebraic modal
logic~\cite{DBLP:journals/ndjfl/Pattinson04,DBLP:journals/tcs/Schroder08}
and show that graded
logics~\cite{DBLP:conf/calco/MiliusPS15,DBLP:conf/concur/DorschMS19}
are naturally viewed as a fragment of coalgebraic modal logic.

Syntactically, a logic is a triple
$\mathcal{L} = (\Theta, \mathcal{O}, \Lambda)$ where $\Theta$ is a set
of truth constants,
$\mathcal{O}$ is a set of propositional operators, each with
associated finite arity, and $\Lambda$ is a set of modal operators,
also each with an associated finite arity. The set of formulae of
$\mathcal{L}$ is given by the grammar
\[\phi ::= \theta \mid p(\phi_1, \ldots, \phi_n) \mid \lambda(\phi_1, \ldots, \phi_m)\]
where $p \in \mathcal{O}$ is $n$-ary, $\lambda\in \Lambda$ is $m$-ary and $\theta \in \Theta$.

Semantically, formulae are interpreted in coalgebras of some functor ${G\colon \Ccat \to \Ccat}$, taking values in a truth-value object $\Omega$ of $\Ccat$. We assume that $\Ccat$ has finite products and a terminal object. The semantics of a formula $\phi$ in a $G$-coalgebra $(C, c)$ is a morphism $\sem{\phi}_c\colon C \to \Omega$. The semantics is parametric in the following components:
\begin{itemize}
	\item For each $\theta \in \Theta$ a $\Ccat$-morphism $\hat{\theta}\colon 1 \to \Omega$.
	\item For each $p \in \mathcal{O}$ with arity $n$ a $\Ccat$-morphism $\sem{p}\colon\Omega^n \to \Omega$
    \item For each $\lambda\in \Lambda$ a $\Ccat$-morphism $\sem{\lambda}\colon G(\Omega^n) \to \Omega$ 
\end{itemize}
The semantics of formulae is then defined inductively:
\begin{itemize}
	\item For $\theta \in \Theta$ we define $\sem{\theta}_c = C \xrightarrow{!} 1 \xrightarrow{\hat{\theta}} \Omega$
	\item For $p \in \mathcal{O}$ we define $\sem{p(\phi_1, \ldots, \phi_n)}_c = \sem{p} \cdot \langle \sem{\phi_1}_c, \ldots, \sem{\phi_n}_c\rangle$
    \item For $\lambda \in \Lambda$ we define $\sem{\lambda(\phi_1, \ldots, \phi_m)}_c = \sem{\lambda} \cdot G\langle\sem{\phi_1}_c, \ldots, \sem{\phi_m}_c\rangle \cdot c$ 
\end{itemize}

\noindent The following definition of logical distance quantifies over
all formulae~$\phi$ of \emph{uniform depth}, meaning that all
occurrences of truth constants in~$\phi$ are under the same number of
nested modal operators. This is a mild restriction; in fact, for the
above version of coalgebraic logic, truth constants can always be
modelled as $0$-ary propositional operators, for which there is no
uniformity restriction. Uniform depth does come to play a role once we
talk about graded logics, where propositional operators are
additionally required to be $gM_0$-algebra homomorphisms, while truth
constants are not. If $M_0$ is affine, then all $\Ccat$-morphisms
$1\to A$ into $M_0$-algebras~$A$ are $M_0$-algebra homomorphisms.

For the rest of the paper, assume that $\Ccat$ is one of $\met_\V$,
$\pmet_\V$ or $\dpmet_\V$; in particular, the truth value
object~$\Omega$ carries the structure of a $\V$-category.

\begin{defn}
  The \emph{logical distance} of states $x, y \in C$ in a
  $G$-coalgebra $(C, c)$ under the logic $\mathcal{L}$
  is
  \begin{equation*}\textstyle
    d^{\mathcal{L}}(x, y) = \bigwedge_{n \in \mathbb{N},
	\phi \in \mathcal{L}_n} d_\Omega(\sem{ \phi}_c(x), \sem{ \phi }_c(y)) 
\end{equation*}
where $\mathcal{L}_n$ is the set of all uniform depth-n $\mathcal{L}$ formulae.
	We say that $\mathcal{L}$ is \emph{invariant} for $(\alpha, \mathbb{M})$ if $d^{b} \leq d^{\mathcal{L}}$ and \emph{expressive} if
	$d^{b} \geq d^{\mathcal{L}}$.
\end{defn}
\noindent It is straightforward to show that the logic defined above
is invariant under behavioural equivalence, i.e. the graded
equivalence induced by $\mathbb{M}_G$
(Example~\ref{ex:gradedMonads}.\ref{item:branching}).  We want to
identify logics that are invariant not only under behavioural
equivalence, but under an arbitrary graded semantics. To this end, we
define graded logics:

\begin{defn}
  Let $(\alpha, \mathbb{M})$ be a graded semantics for $G$ and
  $o\colon M_0\Omega \to \Omega$ an $M_0$-algebra structure on
  $\Omega$. A logic $\mathcal{L}$ is a \emph{graded logic} (for
  $(\alpha, \mathbb{M})$) if the following hold:
\begin{enumerate}
\item For every $n$-ary $p\in \mathcal{O}$, the morphism $\sem{p}$ is an
    $M_0$-algebra homomorphism ${(\Omega, o)^n\to(\Omega, o)}$.
\item The semantics of $\lambda \in \Lambda$ factors as
  $\sem{\lambda} = f \cdot \alpha_{\Omega^n}$ such that the tuple
  $(\Omega^n, \Omega, o^{(n)}, o, f)$ constitutes an
  $M_1$-algebra. More concretely, this means that it satisfies
  $f \cdot \mu^{1,0} = f \cdot M_1o^{(n)}$ (we refer to this property
  as \emph{coequalization}), as well as
  $f \cdot \mu^{0,1} = o \cdot M_0f$ (\emph{homomorphy}), or written
  diagrammatically:
	\begin{equation*}
		\begin{tikzcd}
M_1M_0\Omega^n \arrow[r, "M_1o^{(n)}"', shift right] \arrow[r, "{\mu^{1,0}}", shift left] & M_1\Omega^n \arrow[r, "f"] & \Omega
\end{tikzcd}
\qquad
		\begin{tikzcd}
M_0M_1\Omega^n \arrow[r, "M_0f"] \arrow[d, "{\mu^{0,1}}" left] & M_0\Omega \arrow[d, "o"] \\
M_1\Omega^n \arrow[r, "f" below]                                & \Omega                  
\end{tikzcd}
	\end{equation*}
\end{enumerate}
\end{defn}

\noindent In many examples (including those discussed in this work),
the factorization in Condition~2 is simplified by the fact that
$\alpha = id$, and just requires that
$(\Omega^n, \Omega, o^{(n)}, o, \sem{\lambda})$ is an
$M_1$-algebra. For readability, we restrict the technical development
to unary modalities from now on; treating modalities of arbitrary
arity is simply a matter of adding indices. In examples, modalities
will have arity either~$1$ or~$0$.

\begin{propositionrep}
	\label{prop:invariance}
    Let $\mathcal L$ be a graded logic for the semantics $(\alpha, \mathbb M)$ on $G\colon \Ccat \to \Ccat$ and $(C, c)$ a $G$-coalgebra. For two states $x, y \in C$ we have that ${d^b(x, y) \leq d^{\mathcal{L}}(x, y)}$.
\end{propositionrep}
\begin{proof}
	We define an evaluation of formulae on the semantic objects as morphisms $\sem{\phi}_{\mathbb M}\colon M_n1 \to \Omega$, and show that $\sem{\phi}_c = \sem{\phi}_\mathbb{M} \cdot c^{(n)}$. Let $x, y \in C$ be states of a coalgebra $(C, c)$. The claim then follows from the fact that $d^b(x, y) \leq d_{M_n1}(c^{(n)}(x), c^{(n)}(y))$  and the $\sem{\phi}_\mathbb{M}$ are $\V$-functors. We define the semantics $\sem{\cdot}_\mathbb{M}$:

\begin{itemize}
	\item  $\sem{\theta}_\mathbb{M} = M_01 \xrightarrow{M_0\hat{\theta}} M_0\Omega \xrightarrow{o} \Omega$ for $\theta \in \Theta$
	\item $\sem{ p(\phi_1, \ldots, \phi_n)}_\mathbb{M} = \sem{p} \cdot \langle \sem{ \phi_1}_\mathbb{M}, \ldots, \sem{ \phi_n }_\mathbb{M} \rangle$ for $p \in \mathcal{O}$ $n$-ary
	\item $\sem{ \lambda \phi}_\mathbb{M} =
          f(\sem{\phi}_{\mathbb{M}})$ for $\lambda \in \Lambda$ (we
          continue to restrict to unary modalities)
\end{itemize}
\noindent where~$f$ in the clause for modal operators comes from
canonicity of $(M_n, M_{n+1}, \mu^{0,n}, \mu^{0,n+1}, \mu^{1,n})$
(Proposition~\ref{canonical}), that is,
$f(\langle\sem{\phi_1}_\mathbb{M}, \ldots,
\sem{\phi_m}_\mathbb{M}\rangle)$ is the, by freeness unique, morphism
that makes the following square commute:

\begin{equation}\label{eq:mod}
  \begin{tikzcd}
    M_1M_n1 \arrow[r, "M_1\sem{\phi}_{\mathbb{M}}"] \arrow[d, "\mu^{1n}"] &[20pt] M_1(\Omega) \arrow[d, "f"] \\
    M_{n+1}1 \arrow[r, "f(\sem{\phi}_{\mathbb{M}})"] &[20pt] \Omega
  \end{tikzcd}
\end{equation}
\noindent It is straightforward to show by induction on the depth
of~$\phi$ that the morphism $\sem{\phi}_\mathbb{M}$ defines a
homomorphism of $M_0$-algebras from $(M_n1, \mu^{0,n})$ to $(\Omega, o)$,
which is needed for $\sem{\lambda\phi}_\mathbb{M}$ to be defined.

Now fix a coalgebra $(C, c)$ and a uniform-depth formula $\phi$ of
$\mathcal L$. We prove the claim that
$\sem{\phi}_c = \sem{\phi}_\mathbb{M} \cdot c^{(n)}$ by structural
induction on $\phi$.

For the case of $\phi = \theta \in \Theta$ we have, by unrolling
definitions, that $\sem{\theta}_c = \hat{\theta} \cdot !_X$ and
$ \sem{\phi}_\mathbb{M} \cdot c^{(n)} = o\cdot M_0\hat{\theta} \cdot
M_0!_X \cdot \eta_X$, which are the outer paths in the following
diagram:
\begin{equation*}
  \begin{tikzcd}
    X \arrow[r, "!"] \arrow[d, "\eta_X"] & 1 \arrow[d, "\eta_1"] \arrow[r, "\hat \theta"] & \Omega \arrow[d, "\eta_\Omega"] \arrow[rd, "id"] &        \\
    M_0X \arrow[r, "M_0!"]               & M_01 \arrow[r, "M_0 \hat \theta"]              & M_0\Omega \arrow[r, "o"]                         & \Omega
  \end{tikzcd}
\end{equation*}
\noindent The squares commute due to naturality of $\eta$, while
commutativity of the triangle is implied by $o$ being an
$M_0$-algebra.  The step for formulae of the form
$\phi = p(\phi_1, \ldots, \phi_n)$ is immediate from definitions. For
$\phi = \lambda\phi'$ with $\phi'$ of uniform depth~$n$, we have
\begin{equation*}
  \begin{aligned}
      &\sem{\phi}_c = \sem{\lambda} \cdot G\sem{\phi'}_c \cdot c &\\
      &= f \cdot \alpha_\Omega \cdot G\sem{\phi'}_c \cdot c &\\
    &= f \cdot M_1\sem{\phi'}_c \cdot \alpha_X \cdot c  & \by{naturality of $\alpha$}\\
    &= f \cdot M_1\sem{\phi'}_\mathbb{M}\cdot M_1c^{(n)}\cdot\alpha_X \cdot c & \by{IH}\\
    &=f(\sem{\phi'}_\mathbb{M}) \cdot \mu^{1n} \cdot M_1c^{(n)}\cdot\alpha_X \cdot c & \by{\ref{eq:mod}}\\
    &= \sem{\phi}_\mathbb{M} \cdot c^{(n+1)} &
  \end{aligned}
\end{equation*}
\qed
\end{proof}

\begin{proof}[Sketch]
The proof is based on showing, by induction on~$\phi$, the stronger
property that the evaluation functions $\sem{\phi}_c$ of
depth-$n$ formulae~$\phi$ factor through $M_0$-homomorphisms
\begin{equation}\label{eq:old-semantics}
  \sem{\phi}_\mathbb{M}\colon M_n1\to\Omega,
\end{equation}
as used in earlier formulations of the
semantics~\cite{DBLP:conf/concur/DorschMS19,DBLP:conf/lics/FordMS21},
with canonicity of~$M_n1$ (Lemma~\ref{lem:canonical}) being the key
property in the step for modalities.\qed
\end{proof}

\noindent  The proof uses uniformity to enable the factorization of
  formula evaluation via a single $M_n1$, which in general is possible
  only for uniform-depth formulae. In general, non-uniform depth formulae of graded logics fail to be invariant. We provide an example for this fact in the appendix.\lsnote{replace with reference to full version}
  \begin{toappendix}
    \subsubsection*{Details for failure of invariance of non-uniform depth fomulae}
    As a counterexample, consider the Kleisli-style graded monad, i.e. $M_n = TF^n$ where the monad part
    $T = \mathcal{P}$, functor $FX = X \times X$, and the Kleisli
    distributive law $\zeta\colon FT\to TF$ given by
    $\zeta(A, B) = X \times Y$ for $A, B \subseteq X$. These data
    induce a minimal form of \emph{tree-shaped-trace} semantics: For a
    state~$x$ in a $TF$-coalgebra $(C,c)$,
    $c^{(n)}(x)\in TF^n1=\power 1$ records whether the complete binary
    tree of depth~$n$ can be executed at~$x$. We define a
      graded logic $\mathcal{L}$ over $(\Omega, o)$, where $\Omega = \{\bot, \top\}$
    and $o\colon \mathcal{P}\Omega \to \Omega$ takes suprema. The logic
    contains a truth constant $\top$, where
    $\hat{\top} \colon 1\to \Omega$ is the constant map to $\top$. We also
    have a binary modal operator~$\Diamond$, with
    $\sem{\Diamond}\colon \mathcal{P}(\Omega^2 \times \Omega^2) \to
    \Omega$ defined as
    $\sem{\Diamond}(S) = o(\mathcal{P}\pi_1(S)) \wedge
    o(\mathcal{P}\pi_4(S))$ where $\pi_i$ is the $i$-th projection of
    $\Omega^2\times\Omega^2\cong\Omega^4$; that is,
    $\Diamond(\phi_1,\phi_2)$ evaluates to $\top$, if there is a successor
    pair whose first component satisfies~$\phi_1$ and whose second
    component satisfies~$\phi_2$. We define a $TF$-coalgebra
    $(\{x, y, z\}, c)$ where $c(x) = \{(z, z)\}$,
    $c(y) = \{(x, z)\}$, and $c(z) = \emptyset$. Then~$x$
    and~$y$ disagree on the non-uniform formula
    $\Diamond(\Diamond(\top, \top),\top)$, even though
    $x$ and $y$ are equivalent under the graded semantics.
  \end{toappendix}
  Recall that when~$M_0$ is affine,
  then uniform depth is not an actual restriction.
\noindent Having established invariance, we next generalize the
expressivity criterion for graded
logics~\cite{DBLP:conf/concur/DorschMS19} to our present quantitative
setting: 
\begin{defn}
  \label{def:separation}A graded logic~$\mathcal{L}$ consisting of
  $\Theta$, $\mathcal{O}$, $\Lambda$ is \emph{depth-0 separating} if
  the family of maps $\{ o \cdot M_0\hat\theta \colon M_01\to \Omega\mid c \in \Theta\}$
  is initial.  Moreover,~$\mathcal{L}$ is \emph{depth-1 separating} if
  for all canonical $M_1$-algebras $A$ and initial
  sources $\mathfrak{A}$ of $M_0$-homomorphisms $(A_0, a^{0,0}) \rightarrow (\Omega,o)$,
  closed under the propositional operators in~$\mathcal{O}$,
  the set
  $$\Lambda(\mathfrak{A}) := \{\sem{\lambda}(f): A_1 \rightarrow \Omega \mid
    \lambda \in \Lambda, f \in \mathfrak{A} \}$$ is initial, where $\sem{\lambda}(f)$ is the by canonicity unique morphism such that ${\sem{\lambda}(f) \cdot a^{1,0} = \sem{\lambda}\cdot M_1f}$.
\end{defn}
\noindent Essentially, the above conditions encapsulate what is needed
to push initiality through an induction on the depth of formulae. We
thus obtain
\begin{theoremrep}
  \label{thm:main-expressivity-graded-logics}
  Suppose that a graded logic $\mathcal{L}$ is both
  depth-0 separating and depth-1 separating. Then $\mathcal{L}$ is
  expressive.
\end{theoremrep}

\begin{proof}
  We utilize the semantics $\sem{-}_\mathbb{M}$, defined in the proof
  of Proposition~\ref{prop:invariance}.  It suffices to show that the family of
  maps
  \begin{equation*}
	  \{\sem{\phi}_\mathbb{M}: M_n1 \to \Omega \mid \phi \text{ is
      a uniform depth-$n$  $\mathcal L$ formula}\}
  \end{equation*}
  is initial for each~$n$.  We proceed by induction on $n$. The base
  case $n = 0$ is immediate by depth-$0$ separation. For the inductive
  step, let $\mathfrak{A}$ denote the set of evaluations
  $M_n1 \rightarrow \Omega$ of depth-$n$ formulas. By the induction
  hypothesis, $\mathfrak{A}$ is initial.  By
  definition,~$\mathfrak{A}$ is closed under propositional operators
  in $\mathcal{O}$. 
  By depth-$1$ separation, it follows that set
  \begin{equation*}
    \{\llbracket \lambda \rrbracket(\llbracket \phi \rrbracket) \colon M_{n+1} \to \Omega \mid \lambda \in \Lambda, \phi \text{ a uniform depth-$n$ formula}\}
  \end{equation*}
  is initial, proving the claim.
\end{proof}

\section{Graded Semantics via Coalgebraic Determinization }
\label{sec:DetGradedSem}

From now on, fix a  $\Ccat$-endofunctor~$F$, a monad $T$ on $\Ccat$,
and an EM law $\zeta\colon TF\Rightarrow FT$. The objective of this
section is to show that behavioural equivalences, respectively
metrics, on a determinized coalgebra agree with the
equivalences/metrics induced by the graded semantics
(Lemma~\ref{lem:CharTracesFinBehDistance}), and that graded logics for
$FT$ may be reduced to coalgebraic logics for~$F$. We recall the
notion of predeterminization under a graded
semantics~\cite{DBLP:conf/lics/FordMSB022} and show that this is the
same concept as determinization under an EM law, under the
condition that the monad $T$ is affine.


Let $\mathbb{M}$ be a graded monad. We have a functor
$E\colon \Galg{0}{\mathbb{M}} \to \Galg{1}{\mathbb{M}}$ that takes an
$M_0$-algebra~$A$ to the free $M_1$-algebra over~$A$ with respect to
$({-})_0$ (which is then
canonical, cf. Definition~\ref{def:canonical}). This gives rise to a functor
\begin{equation*}
  \mbar= (\Galg{0}{\mathbb{M}} \xrightarrow{E}
  \Galg{1}{\mathbb{M}} \xrightarrow{({-})_1} \Galg{0}{\mathbb{M}}),
\end{equation*}
which intuitively takes an $M_0$-algebra of behaviours to the
$M_0$-algebra of behaviours having absorbed one more step.  Since
$(M_0X, M_1X, \mu^{0,0}_X, \mu^{0,1}_X, \mu^{1,0}_X)$ is canonical
(Proposition~\ref{canonical}), we have
$\mbar(M_0X, \mu^{0,0}) = (M_1X, \mu^{0,1})$, or stated slightly
differently, if we denote the free-forgetful adjunction on
$\Galg{0}{\mathbb{M}}$ by $L \dashv R$, then $M_1 = R \mbar L$.  For a
graded semantics $(\alpha\colon G \to M_1, \mathbb{M})$ and a coalgebra
$c\colon C \to GC$, we have
$C \xrightarrow{\alpha \cdot c} M_1C = R\mbar LC$. The adjunction then
yields a unique morphism $ c^{\dagger}\colon LC \to \mbar LC$,
defining a form of determinization under the graded semantics, similar
to the generalized powerset construction. Specifically, if $M_01 = 1$,
then for $x, y \in C$, $\eta(x)$ and $\eta(y)$ are behaviourally
equivalent in $c^\dagger$ iff~$x$ and~$y$ are identified under the
graded semantics $(\alpha, \mathbb{M})$. We show next that

\begin{lem}
  If $\mathbb M = \mathbb M_\zeta$ then $\mbar = \tilde{F}$.
\end{lem}
\begin{proof}
	Let $(A, a)$ be a $T$-algebra. Then $\tilde{F}(A, a) = (Fa, Fa \cdot \zeta_A)$. On the other hand, by Lemma~\ref{lem:canonical}, the 1-part of the canonical algebra of $\mbar(A, a)$ is given by the following (split) coequalizer:
\begin{equation*}
\begin{tikzcd}
FTTA \arrow[r, "F\mu_A", shift left] \arrow[r, "FTa"', shift right] & FTA \arrow[r, "Fa"] \arrow[l, "F\eta_{TA}", bend left=60] & FA \arrow[l, "F\eta_A", bend left=49]
\end{tikzcd}
\end{equation*}
Commutativity of all relevant paths is obvious from the algebra and monad axioms, implying that the diagram is a
	coequalizer diagram by virtue of being a split coequalizer. Then $(A, FA, a, Fa \cdot \zeta_A, Fa)$ defines a canonical $M_1$-algebra where coequalization, as well as canonicity (due to Lemma~\ref{lem:canonical}), are by the above coequalizer, and homomorphy instantiates to the outer paths of the following diagram:
\begin{equation*}
\begin{tikzcd}
TFTA \arrow[r, "\zeta_{TA}"] \arrow[d, "TFa"] & FTTA \arrow[r, "F\mu"] \arrow[d, "FTa"] & FTA \arrow[d, "Fa"] \\
TFA \arrow[r, "\zeta_A"]                      & FTA \arrow[r, "Fa"]                     & FA
\end{tikzcd}
\end{equation*}
	Commutativity of the outer rectangle follows from the fact that the left square commutes by naturality of $\zeta$ and the right square commutes by virtue of $(A,a)$ being a $T$-algebra. Taking the 1-part of this canonical algebra then leaves us with $(Fa, Fa \cdot \zeta_A)$.
	On morphisms $h\colon (A, a) \to (B, b)$, the lifting $\tilde{F}$ acts by sending $h$ to $Fh$. Commutativity of the relevant diagram making $Fh$ a $T$-algebra morphism between $FA$ and $FB$ is easily checked, as is the fact that $(h, Fh)$ constitutes a morphism between the canonical $M_1$-algebras.\qed
\end{proof}

\begin{lem}
	Let $(C, c)$ be an $FT$-coalgebra and $c^\dagger$ the predeterminization under the graded semantics $M_\zeta$. Then $c^\# = c^\dagger$.
\end{lem}
\begin{proof}
	This follows from the fact that $c^\#$ can equivalently be defined as the adjoint transpose of $c$ under the free-forgetful adjunction of $\alg{T}$ \cite{DBLP:conf/fsttcs/SilvaBBR10}. Then $c^\#$ and $c^\dagger$ agree by definition.\qed
\end{proof}

\begin{defn}
	Let $T\colon \Ccat \to \Ccat$ be a monad and $H\colon \alg{T}\to \alg{T}$ a functor on the corresponding Eilenberg-Moore category. Further, let $c\colon ((A, d_A),a) \to H((A, d_A), a)$ be an $H$-coalgebra. The \emph{finite-depth behavioural distance} of two states $x, y \in A$ is given by $d^{H}(x, y) = \bigwedge_{i\in \nat} d_{H^i1}(f_i(x), f_i(y))$, where the $f_i\colon A\to H^i 1$ are the projections into the final $H$-chain.\todo{P: find a better notation for this, maybe $d^H_\omega$?}
\end{defn}

\begin{lem}\label{lem:CharTracesFinBehDistance}
  Let $(\alpha\colon G \to M_1, \mathbb{M})$ be a graded semantics on
  $\Ccat$ with~$M_0$ affine, and let $(C, c)$ be a
  $G$-coalgebra. Then $d^{\mbar}(\eta(x), \eta(y)) = d^b(x, y)$ for
  all $x, y \in C$.
\end{lem}
\begin{proof}
  One shows by induction on~$n$ that $f_n \cdot \eta = c^{(n)}$ for
  all~$n$, where affinity is needed for the base case $n=0$.\qed
\end{proof}

\begin{rem}
  In the case where the graded monad is $\mathbb{M}_\zeta$, if $T$ is
  affine, then the final chain of $\mbar$ lives over the final
  chain of $F$. In particular, if $F$ is finitary and $\V=2$, then
  finite-depth behavioural equivalence agrees with behavioural
  equivalence, for both $F$-coalgebras and $\mbar$-coalgebras.
\end{rem}

\begin{rem}
  As noted in Example~\ref{ex:gradedMonads}.\ref{item:emgm}, finite-depth
  behavioural distance in $\alg{T}$ may be coarser than the graded
  semantics but may then be canonically recovered from the graded semantics. 
\end{rem}

\noindent From now on, we notationally conflate modalities $\lambda\in\Lambda$
and their interpretations $\sem{\lambda}\colon FT\Omega\to\Omega$.
The following result completely characterizes the modal operators of
graded logics for the semantics $(\id, \mathbb{M}_\zeta)$:

\begin{theoremrep}
  \label{thm:validOperators}
  Let $\lambda\colon FT\Omega \to \Omega$ be a modal operator for a
  graded logic with truth value object $(\Omega, o)$. Then
  $\lambda = ev_\lambda \cdot Fo$ for some algebra homomorphism
  $ev_\lambda \colon\tilde{F}(\Omega, o) \to (\Omega, o)$. On the
  other hand, every algebra homomorphism
  $\tilde{F}(\Omega, o) \to (\Omega, o)$ yields a modal operator in
  this way.
\end{theoremrep}
\begin{proof}
  Since $\lambda$ is an $M_1$-algebra structure and thus satisfies the
  coequalization property, it factors through the coequalizer of
  $\mu^{1,0}_{\Omega}\colon M_1M_0\Omega \to M_1\Omega$ and
  $M_1o\colon M_1M_0\Omega \to M_1\Omega $, which, by definition, is
  given by $\mbar{(\Omega, o)} = (F\Omega, Fo\cdot\zeta_{\Omega})$, as
  displayed in the following diagram:
  \begin{equation*}
    \begin{tikzcd}
      FTT\Omega \arrow[r, "F\mu_{\Omega}", shift left] \arrow[r, "FTo"', shift right] & FT\Omega \arrow[rd, "Fo"] \arrow[rr, " \lambda"] &                           & \Omega \\
      & & F\Omega \arrow[ru, "ev_\lambda"] &
\end{tikzcd}
\end{equation*}

\noindent To show that $ev_\lambda$ is a homomorphism of $T$-algebras consider the following diagram:\jfnote{This part is not necessary if we can argue that the coequalizer is taken in Eilenberg-Moore}

\begin{equation*}
\begin{tikzcd}
FT\Omega \arrow[rd, "Fo"] \arrow[rr, " \lambda "]                                    &                                                            & \Omega                    \\
																								   & F\Omega \arrow[ru, "ev_\lambda"]                                  &                           \\
FTT\Omega \arrow[uu, "F\mu"] \arrow[r, "FTo"]                                                      & FT\Omega \arrow[u, "Fo"']                          &                           \\
																								   & TF\Omega \arrow[rd, "Tev_\lambda"] \arrow[u, "\zeta_{\Omega}"'] &                           \\
TFT\Omega \arrow[rr, "T\lambda "] \arrow[ru, "TFo"] \arrow[uu, "\zeta_{T\Omega}"] &                                                            & T\Omega \arrow[uuuu, "o"]
\end{tikzcd}
\end{equation*}
The outer square commutes by homomorphy of $\lambda$. The top
triangle commutes by coequalization, as does the bottom triangle,
since it is just $T$ applied to the top triangle. The left top square
commutes since $o$ is a $T$-algebra structure and the left bottom
square commutes by naturality of $\zeta$. It follows that the right
hand square precomposed with $TFo$ commutes. $TFo$ is a split
coequalizer, and therefore an epimorphism. Therefore by canceling
$TFo$ we have that the right hand square commutes, which is precisely
the condition for $ev_\lambda$ to be a homomorphism of $T$-algebras.

Conversely, given a morphism $ev_\lambda \colon \tilde F(\Omega, o) \to (\Omega, o)$ it is straight forward to check that $\lambda =ev_\lambda \cdot Fo$ satisfies the laws necessary to make it an $M_1$-algebra main structure.\qed
\end{proof}
\noindent As our second main result, we next show that a logic is
depth-1 separating for the semantics of $\mathbb{M}_\zeta$ if the
$F$-algebra part of its modal operators is expressive for $F$. This
criterion is typically very easy to establish and can be shown for
general classes of functors, which is what we mean by our slogan that
expressive graded logics for EM semantics come essentially for free.

\begin{thm}
  \label{thm:main_F_separation}
  Let $\mathcal{L} = (\Theta, \mathcal{O}, \Lambda)$ be a graded logic
  for $\mathbb{M}_\zeta$ and
  $\mathcal{L}' = (\Theta, \mathcal{O}, \Lambda')$ the (graded) logic
  for $\mathbb{M}_F$ with
  $\Lambda'=\{f \colon F\Omega \to \Omega\mid f \cdot Fo \in
  \Lambda\}$. Then~$\mathcal{L}$ is depth-1 separating for
  $\mathbb{M}_\zeta$ if~$\mathcal{L}'$ is depth-1 separating for
  $\mathbb{M}_F$.
\end{thm}
\begin{proof}
  Let $A$ be a canonical $M_1$-algebra. Since $\mbar = \tilde{F}$, we
  know that $A$ has the form
  $(A_0, FA_0, a^{0,0}, Fa^{0,0}\cdot \zeta, Fa^{0,0})$. For a
  homomorphism $h\colon (A_0, a^{0,0}) \to (\Omega, o)$ of $T$-algebras
  and $\lambda \in \Lambda$ where $\lambda = f \cdot Fo$, $\lambda(h)$
  is, by definition, the unique morphism that makes the outer
  rectangle in the following diagram commute:
  \begin{equation*}
    \begin{tikzcd}
      FTA_0 \arrow[r, "FTh"] \arrow[dd, "Fa^{0,0}"'] & FT\Omega \arrow[d, "Fo"] \\
      & F\Omega \arrow[d, "f"]      \\
      FA_0 \arrow[r, "\lambda(h)"', dashed] \arrow[ru, "Fh"]    & \Omega
    \end{tikzcd}
  \end{equation*}
  \noindent The top square commutes since it is just $F$ applied to
  the homomorphism square of $h$. Since~$a^{0,0}$ is a split epimorphism
  (by virtue of being an algebra for a monad), $Fa^{0,0}$ is also a
  split epimorphism. Therefore, $\lambda(h) = f \cdot Fh$. Let
    ${\mathfrak{A} \subseteq \alg{T}((A_0, a^{0,0}), (\Omega, o))}$ be a
  separating set of algebra homomorphisms; we have to show that
  $\Lambda(\mathfrak{A})$ is separating. But since
  $\lambda(h) = f \cdot Fh$ for all $h\in\mathfrak{A}$ and
  $\lambda=f\cdot Fo\in\Lambda$, we have
  $\Lambda'(\mathfrak{A}) = \{f \cdot Fh \mid f\in \Lambda', h \in
  \mathfrak{A}\}=\Lambda(\mathfrak{A})$, and $\Lambda'(\mathfrak{A})$
  is spearating by depth-$1$ separation for $\mathcal{L}'$.\qed
\end{proof}

\section{Examples}
\label{sec:examples}

In our central examples,~$F$ takes the form $\V \times (-)^\Sigma$
while $T$ varies. In these cases, we always have a set of separating
modalities: We have the set
${\Lambda' = \{ \ev{ \sigma } \mid \sigma \in \Sigma \} \cup \{
    \ev{\top} \}}$ of modalities for~$F$, where
$\ev{\sigma}\colon \V \times \V^\Sigma \to \V$ is a unary operator
defined by $(v, f) \mapsto f(\sigma)$, and
$\ev{\top}\colon \V \times 1^\Sigma \to \V$ is the $0$-ary operator
defined by $(v, f) \mapsto v$. For a monad $T$ and an algebra
structure $o\colon T\V \to \V$, the semantics of each
$\ev{\lambda} \in \Lambda'$ extends to a modal operator
$\modal{\lambda}$ for $FT$, given by
$\modal{\lambda} = \ev{\lambda}\cdot Fo$. We thus have coalgebraic
logics $\mathcal{L}' = (\emptyset, \emptyset, \Lambda')$ for $F$ and
$\mathcal{L} = (\emptyset, \emptyset, \Lambda)$ for $FT$.

\begin{lem}
  Let $F = \V \times {-}^\Sigma$. Let $T$ be a monad and
  $\tilde{F} \colon \alg{T} \to \alg{T}$ a lifting of~$F$. Moreover,
  suppose that $\V$ carries a $T$-algebra structure
  $o\colon T\V \to \V$. Then, for every $\ev{\lambda}\in \Lambda'$,
  the semantics $\ev{\lambda}$ is a homomorphism of algebras
  $\tilde F(\V,o) \to (\V, o)$
\end{lem}
\begin{proof}
  Since the $\ev{\lambda}$ are just product
  projections, this follows from the fact that the forgetful functor
  $U \colon \alg{T} \to \textbf{C}$ creates limits~\cite[Proposition
	20.12]{AdamekHerrlich90}.\qed
\end{proof}
\begin{cor}
	Let $\zeta$ be defined as in Example~\ref{ex:distributive-machine}.
	The logic $\mathcal{L}$ as defined above is a graded logic for the graded semantics $(\id, \mathbb{M}_\zeta)$.
\end{cor}
\begin{lem}
	\label{lem:L_separating_for_machines}
	The logic $\mathcal{L}'$ as defined above is depth-1 separating for the graded semantics $(\id, \mathbb{M}_F)$.
\end{lem}
\begin{proof}
  By Proposition~\ref{canonical}, canonical $M_1$-algebras have the form
  $A = (A_0, FA_0, \id, \id, \id)$. Let $A$ be a canonical
  $M_1$-algebra and $\mathfrak{A}$ an initial source $A_0 \to \V$; we
  then need to show that the lower edges in the following diagram
  collectively form an initial source, where $f$ ranges over
  $\mathfrak{A}$:
  \begin{equation*}
    \begin{tikzcd}
      FA_0 \arrow[d, "\id"'] \arrow[r, "Ff"] & F\mathcal{V} \arrow[d, "\ev \lambda"] \\
      FA_0 \arrow[r, "\ev \lambda(f)", dashed]    & \mathcal{V}
    \end{tikzcd}
  \end{equation*}
  Since the modal operators $\ev \lambda$ are precisely the
  projections of the product $F\V$, they constitute an initial source;
  moreover, again since~$F$ is a product, it preserves initial
  sources, so the source of all $Ff$ is initial. Thus,
  $\Lambda(\mathfrak{A})$ is a composite of initial sources, hence
  itself
  initial. 
  \qed
\end{proof}

\noindent Words in $\Sigma^*$ can be viewed as formulae of
$\mathcal{L}$ in the obvious way. The evaluation $\sem{\phi}_c$ then
captures the notion of acceptance in the automaton given by the
$FT$-coalgebra $(C, c)$. Logics in general however allow to express
far more interesting statements, since on the one hand formulae may
specify words only up to a suffix, and on the other hand the logic may
include propositional operators.  We consider a few concrete examples:

\begin{expl}[Deadlock-free nondeterministic automata]
  \label{ex:nd-automata}
    We take ${\V = \textbf{2}}$, and work in the category
  $\met_\textbf{2}$ of sets and functions. Concretely, this means that
  all objects carry the discrete equivalence relation, and initiality
  of a source is joint injectivity. If $T$ is the nonempty powerset monad
  $\mathcal{P}^+$, then coalgebras
  $c\colon C \to 2 \times (\mathcal{P}^+C)^\Sigma$ are deadlock-free
  nondeterministic automata.  With the algebra structure
  $o\colon \mathcal{P}^+2 \to 2$ defined by $o(X) = \top$ if
  $\top \in X$ and $o(X) = \bot$ otherwise, we can construct a
  distributive law $\zeta$ as in Example~\ref{ex:distributive-machine}.
  Since, by Lemma~\ref{lem:L_separating_for_machines}, $\mathcal{L}$ is
  depth-1 separating for $2 \times {-}^\Sigma$, we have that by
  Theorem~\ref{thm:main_F_separation} the logic $\mathcal{L}$ is expressive
  for $(id, \mathbb{M}_\zeta)$. We can add disjunction as a
  propositional operator, preserving invariance of the logic, since
  disjunction preserves joins (i.e.\ is a homomorphism of
  $\mathcal{P}^+$-algebras).
\end{expl}

\begin{expl}[Reactive probabilistic automata]
  For $\V = [0,1]_\oplus$, we consider reactive probabilistic
  automata. Let $T$ be the (finitely supported) probability
  distribution monad $\dist$ on $\pmet_{[0,1]_\oplus}$, which equips
  the set of distributions with the Kantorovich metric
  (e.g.~\cite{DBLP:conf/fsttcs/BaldanBKK14}). We put $\Omega = [0,1]$,
  equipped with the symmetrized metric $d(x,y)=|x-y|$. We have an
  algebra $o\colon \dist[0,1] \to [0,1]$ taking expected values:
  $o(\mu) = \sum_{v\in [0,1]} v\mu(v)$.  The construction in
  Example~\ref{ex:distributive-machine} then yields a semantics where,
  intuitively, the first component of $F$ determines the probability
  of a state to accept. Upon reading a letter~$a$, the automaton moves
  to a random successor state according to the probability
  distribution on states associated with~$a$. The evaluation
  $\sem{\phi}_c(x)$ is then the expected probability of the state
  $x \in C$ of an automaton $(C, c)$ accepting the word corresponding
  to $\phi$.  The distance of two states $x, y \in C$ is the supremum
  in difference of acceptance across all words in $\Sigma^*$.  Again
  we have expressivity of $\mathcal{L}$ by combining
  Lemma~\ref{lem:L_separating_for_machines} and
  Theorem~\ref{thm:main_F_separation}. The logic remains invariant w.r.t.\
  the semantics when extended with propositional operators that are
  homomorphisms $[0,1]^n\to[0,1]$ of $\dist$-algebras, which in this
  case means they are affine maps, such as convex combinations or
  fuzzy negation $x\mapsto 1-x$.

\end{expl}

\begin{expl}[Reactive probabilistic automata with black hole
  termination] Going beyond the leading example
	$F=\V\times(-)^\Sigma$, we add explicit failure in the vein of \cite{DBLP:journals/lmcs/SokolovaW18} to reactive
  probabilistic automata: We now take $\V=2$, and again view
  $\met_\textbf{2}$ as the category of sets and functions
  (Example~\ref{ex:nd-automata}). Let $\Omega = [0,1]$, equipped with the
  $\dist$-algebra structure $o\colon \dist[0,1] \to [0,1]$ that takes
  expected values. We obtain a distributive law
  $\dist(2 \times {-}+1)^\Sigma \Rightarrow 2 \times
    ((\dist{-})+1)^\Sigma$ by composing the distributive law from
  Example~\ref{ex:distributive-machine} with the law
  $\lambda\colon \dist(- + 1) \Rightarrow (\dist-) + 1$ that maps
  $\mu \in \dist(X+1)$ to $*$ iff $\mu(*) \not = 0$, and to $\mu$
  otherwise, where~$*$ denotes the unique element of~$1$. The
  semantics for this type of automaton is like that of probabilistic
  automata, with the exception that if a run leads to the ``state''
  $*$ with non-zero probability, then the automaton immediately gets
  stuck and rejects the word. For the logic, we consider the same
  operators as in the previous examples, with the modification that
  $\ev{\sigma}(v, f) = \bot$ if $f(\sigma) = *$.  Additionally we
  introduce the modal operator $\ev{\bar{\sigma}}$, which carries the
  semantics $\ev{\bar{\sigma}} (v, f) = \top$ if $f(\sigma) = *$ and
  $\ev{\bar{\sigma}} (v, f) = \bot$ otherwise. It is straightforward
  to check that these operations define $\dist$-algebra homomorphisms,
  making them valid modalities according to Theorem~\ref{thm:validOperators}.

  To verify expressivity, it is sufficient by
  Theorem~\ref{thm:main_F_separation} to prove separation of elements of
  $2 \times (X+1)^\Sigma$, so let $\mathfrak{A}$ be an initial set of
  morphisms of type $X \to 2$. Given $s, t \in 2 \times (X+1)^\Sigma$
  such that $d(s, t) = \bot$, we need to find $\ev\lambda$ and
  $h\in \mathfrak{A}$ (or just $\ev\lambda$ if~$\ev\lambda$ is
  $0$-ary) such that $\ev \lambda(h)(s) \not = \ev \lambda(h)(t)$. If
  $s = (v, f)$ and $t = (w, g)$ differ in their first component
  $v\not = w$, we can choose $\ev \top$. If the elements differ in one
  of the other components $\sigma$, we distinguish cases: If
  $x = f(\sigma) \not = * \not = g(\sigma) = y$, then there is
  $h\in \mathfrak{A}$ separating $x$ from $y$, thus
  $\ev \sigma (h)(x) \not = \ev \sigma(h)(y)$.  Otherwise, if
  $f(\sigma) = * \not = g(\sigma)$, we can choose $\ev{ \bar{\sigma}}$
  to separate~$s$ and~$t$, and similarly for the symmetric case. We
  thus obtain expressiveness in the two-valued sense, i.e.\ the logic
  distinguishes non-equivalent states. Like in the previous example,
  the logic remains invariant when extended with propositional
  operators that are affine maps $[0,1]^n\to[0,1]$.
\end{expl}

\section{Conclusion}
\label{sec:conclusion}

We have discussed characteristic logics for system semantics arising
via determinization in the coalgebraic powerset construction,
so-called Eilenberg-Moore semantics, which relies on a distributive
law of a functor representing the language type of a system over a
monad representing the branching
type~\cite{DBLP:conf/fsttcs/SilvaBBR10}. Leading examples are
languages semantics for various forms of automata. As our main
technical tool, we have exploited that Eilenberg-Moore semantics may
be cast as an instance of graded semantics, which provides generic
mechanisms for designing invariant modal logics and establishing their
expressiveness. Our first main result establishes an overview of all
graded modalities available for Eilenberg-Moore semantics, showing
that these are canoincally obtained from modalities for the language
type and a single modality for the branching type. Our second main
result shows that expressivity of such a logic follows from
branching-time expressivity of the same collection of operators with
respect to the language type. Our results are stated in quantalic
generality, allowing for instantiation to both two-valued and
quantitative types of semantics and logics.

An important next step in the programme of developing graded logics
into a verification framework is the question of how graded semantics
relates to fixpoint logics. While we have focused on Eilenberg-Moore
semantics in the present work, graded semantics does also subsume
Kleisli-style trace semantics~\cite{HasuoEA07}, which poses additional
challenges for the design of characteristic modal logics, in particular
in the quantitative setting.

\bibliographystyle{splncs04}
\bibliography{references}

\begin{thebibliography}{10}
\providecommand{\url}[1]{\texttt{#1}}
\providecommand{\urlprefix}{URL }
\providecommand{\doi}[1]{https://doi.org/#1}

\bibitem{AdamekHerrlich90}
Ad{\'a}mek, J., Herrlich, H., Strecker, G.: Abstract and Concrete Categories.
  Wiley Interscience (1990), available as {\emph{Reprints Theory Appl.\ Cat.}}
  17 (2006), pp. 1-507

\bibitem{DBLP:conf/fsttcs/BaldanBKK14}
Baldan, P., Bonchi, F., Kerstan, H., K{\"{o}}nig, B.: Behavioral metrics via
  functor lifting. In: Raman, V., Suresh, S.P. (eds.) Foundation of Software
  Technology and Theoretical Computer Science, {FSTTCS} 2014. LIPIcs, vol.~29,
  pp. 403--415. Schloss Dagstuhl -- Leibniz-Zentrum f{\"{u}}r Informatik
  (2014). \doi{10.4230/LIPIcs.FSTTCS.2014.403}

\bibitem{bbkk:coalgebraic-behavioral-metrics}
Baldan, P., Bonchi, F., Kerstan, H., K{\"o}nig, B.: Coalgebraic behavioral
  metrics. Logical Methods in Computer Science  \textbf{14}(3) (2018), selected
  papers of the 6th Conference on Algebra and Coalgebra in Computer Science,
  CALCO 2015

\bibitem{bgkm:hennessy-milner-galois}
Beohar, H., Gurke, S., K{\"{o}}nig, B., Messing, K.: {H}ennessy-{M}ilner
  theorems via {G}alois connections. In: Klin, B., Pimentel, E. (eds.) Computer
  Science Logic, {CSL} 2023. LIPIcs, vol.~252, pp. 12:1--12:18. Schloss
  Dagstuhl -- Leibniz-Zentrum f{\"{u}}r Informatik (2023).
  \doi{10.4230/LIPIcs.CSL.2023.12}

\bibitem{DBLP:conf/stacs/bgk24}
Beohar, H., Gurke, S., K{\"{o}}nig, B., Messing, K., Forster, J.,
  Schr{\"{o}}der, L., Wild, P.: Expressive quantale-valued logics for
  coalgebras: an adjunction-based approach. In: Kupferman, O., Beyersdorff, O.,
  Lokshtanov, D., Kant{\'{e}}, M.M. (eds.) Theoretical Aspects of Computer
  Science, {STACS} 2024. LIPIcs, Schloss Dagstuhl -- Leibniz-Zentrum f{\"{u}}r
  Informatik (2024), to appear

\bibitem{BonchiEA12}
Bonchi, F., Bonsangue, M., Caltais, G., Rutten, J., Silva, A.: Final semantics
  for decorated traces. In: Mathematical Foundations of Programming Semantics,
  MFPS 2012. ENTCS, vol.~286, pp. 73--86. Elsevier (2012).
  \doi{10.1016/j.entcs.2012.08.006}

\bibitem{DBLP:journals/fuin/Cirstea17a}
C{\^{\i}}rstea, C.: From branching to linear time, coalgebraically. Fundam.
  Informaticae  \textbf{150}(3-4),  379--406 (2017). \doi{10.3233/FI-2017-1474}

\bibitem{DBLP:conf/concur/DorschMS19}
Dorsch, U., Milius, S., Schr{\"{o}}der, L.: Graded monads and graded logics for
  the linear time - branching time spectrum. In: Fokkink, W.J., van Glabbeek,
  R. (eds.) Concurrency Theory, {CONCUR} 2019. LIPIcs, vol.~140, pp.
  36:1--36:16. Schloss Dagstuhl -- Leibniz-Zentrum f{\"{u}}r Informatik (2019).
  \doi{10.4230/LIPIcs.CONCUR.2019.36}

\bibitem{DBLP:conf/lics/FordMS21}
Ford, C., Milius, S., Schr{\"{o}}der, L.: Behavioural preorders via graded
  monads. In: Logic in Computer Science, {LICS} 2021. pp. 1--13. {IEEE} (2021).
  \doi{10.1109/LICS52264.2021.9470517}

\bibitem{DBLP:conf/lics/FordMSB022}
Ford, C., Milius, S., Schr{\"{o}}der, L., Beohar, H., K{\"{o}}nig, B.: Graded
  monads and behavioural equivalence games. In: Baier, C., Fisman, D. (eds.)
  Logic in Computer Science, LICS 2022,. pp. 61:1--61:13. {ACM} (2022).
  \doi{10.1145/3531130.3533374}

\bibitem{DBLP:conf/csl/Forster0HNSW23}
Forster, J., Goncharov, S., Hofmann, D., Nora, P., Schr{\"{o}}der, L., Wild,
  P.: Quantitative hennessy-milner theorems via notions of density. In: Klin,
  B., Pimentel, E. (eds.) Computer Science Logic, {CSL} 2023. LIPIcs, vol.~252,
  pp. 22:1--22:20. Schloss Dagstuhl -- Leibniz-Zentrum f{\"{u}}r Informatik
  (2023). \doi{10.4230/LIPIcs.CSL.2023.22}

\bibitem{g:linear-branching-time}
van Glabbeek, R.: The linear time -- branching time spectrum~{I}. In: Bergstra,
  J., Ponse, A., Smolka, S. (eds.) Handbook of Process Algebra, chap.~1, pp.
  3--99. Elsevier (2001)

\bibitem{HasuoEA07}
Hasuo, I., Jacobs, B., Sokolova, A.: Generic trace semantics via coinduction.
  Log.\ Meth.\ Comput.\ Sci.  \textbf{3} (2007)

\bibitem{DBLP:conf/icalp/HennessyM80}
Hennessy, M., Milner, R.: On observing nondeterminism and concurrency. In:
  de~Bakker, J.W., van Leeuwen, J. (eds.) Automata, Languages and Programming,
  ICALP 1980. LNCS, vol.~85, pp. 299--309. Springer (1980).
  \doi{10.1007/3-540-10003-2\_79}

\bibitem{DBLP:journals/entcs/Jacobs04a}
Jacobs, B.: Trace semantics for coalgebras. In: Ad{\'{a}}mek, J., Milius, S.
  (eds.) Coalgebraic Methods in Computer Science, {CMCS} 2004. ENTCS, vol.~106,
  pp. 167--184. Elsevier (2004). \doi{10.1016/j.entcs.2004.02.031}

\bibitem{DBLP:conf/cmcs/Jacobs16}
Jacobs, B.: Affine monads and side-effect-freeness. In: Hasuo, I. (ed.)
  Coalgebraic Methods in Computer Science, {CMCS} 2016. LNCS, vol.~9608, pp.
  53--72. Springer (2016). \doi{10.1007/978-3-319-40370-0\_5}

\bibitem{DBLP:journals/jcss/Jacobs0S15}
Jacobs, B., Silva, A., Sokolova, A.: Trace semantics via determinization. J.\
  Comput.\ Syst.\ Sci.  \textbf{81}(5),  859--879 (2015).
  \doi{10.1016/j.jcss.2014.12.005}

\bibitem{DBLP:journals/corr/KerstanK13}
Kerstan, H., K{\"{o}}nig, B.: Coalgebraic trace semantics for continuous
  probabilistic transition systems. Log. Methods Comput. Sci.  \textbf{9}(4)
  (2013). \doi{10.2168/LMCS-9(4:16)2013}

\bibitem{KissigKurz10}
Kissig, C., Kurz, A.: Generic trace logics (2011), arXiv preprint 1103.3239

\bibitem{KlinRot15}
Klin, B., Rot, J.: Coalgebraic trace semantics via forgetful logics. In:
  Foundations of Software Science and Computation Structures, FoSSaCS 2015.
  LNCS, Springer (2015)

\bibitem{kkkrh:expressivity-quantitative-modal-logics}
Komorida, Y., Katsumata, S., Kupke, C., Rot, J., Hasuo, I.: Expressivity of
  quantitative modal logics : Categorical foundations via codensity and
  approximation. In: Logic in Computer Science, {LICS} 2021. pp. 1--14. {IEEE}
  (2021). \doi{10.1109/LICS52264.2021.9470656}

\bibitem{DBLP:conf/concur/KonigM18}
K{\"{o}}nig, B., Mika{-}Michalski, C.: (metric) bisimulation games and
  real-valued modal logics for coalgebras. In: Schewe, S., Zhang, L. (eds.)
  Concurrency Theory, {CONCUR} 2018. LIPIcs, vol.~118, pp. 37:1--37:17. Schloss
  Dagstuhl -- Leibniz-Zentrum f{\"{u}}r Informatik (2018).
  \doi{10.4230/LIPIcs.CONCUR.2018.37}

\bibitem{DBLP:journals/lmcs/KupkeR21}
Kupke, C., Rot, J.: Expressive logics for coinductive predicates. Log.\ Methods
  Comput.\ Sci.  \textbf{17}(4) (2021). \doi{10.46298/lmcs-17(4:19)2021}

\bibitem{DBLP:conf/birthday/KurzMPS15}
Kurz, A., Milius, S., Pattinson, D., Schr{\"{o}}der, L.: Simplified coalgebraic
  trace equivalence. In: Nicola, R.D., Hennicker, R. (eds.) Software, Services,
  and Systems -- Essays Dedicated to Martin Wirsing on the Occasion of His
  Retirement from the Chair of Programming and Software Engineering. LNCS,
  vol.~8950, pp. 75--90. Springer (2015). \doi{10.1007/978-3-319-15545-6\_8}

\bibitem{lawvere1973metric}
Lawvere, F.W.: Metric spaces, generalized logic, and closed categories.
  Rendiconti del seminario mat{\'e}matico e fisico di Milano  \textbf{43},
  135--166 (1973)

\bibitem{DBLP:journals/jcss/MartiV15}
Marti, J., Venema, Y.: Lax extensions of coalgebra functors and their logic. J.
  Comput. Syst. Sci.  \textbf{81}(5),  880--900 (2015).
  \doi{10.1016/j.jcss.2014.12.006}

\bibitem{DBLP:conf/calco/MiliusPS15}
Milius, S., Pattinson, D., Schr{\"{o}}der, L.: Generic trace semantics and
  graded monads. In: Moss, L.S., Sobocinski, P. (eds.) Algebra and Coalgebra in
  Computer Science, {CALCO} 2015. LIPIcs, vol.~35, pp. 253--269. Schloss
  Dagstuhl -- Leibniz-Zentrum f{\"{u}}r Informatik (2015).
  \doi{10.4230/LIPIcs.CALCO.2015.253}

\bibitem{MULRY2002:LiftingAlgebras}
Mulry, P.: Lifting results for categories of algebras. Theoretical Computer
  Science  \textbf{278}(1),  257--269 (2002).
  \doi{https://doi.org/10.1016/S0304-3975(00)00338-8}, mathematical Foundations
  of Programming Semantics 1996

\bibitem{DBLP:journals/ndjfl/Pattinson04}
Pattinson, D.: Expressive logics for coalgebras via terminal sequence
  induction. Notre Dame J.\ Formal Log.  \textbf{45}(1),  19--33 (2004).
  \doi{10.1305/ndjfl/1094155277}

\bibitem{DBLP:journals/logcom/RotJL21}
Rot, J., Jacobs, B., Levy, P.B.: Steps and traces. J. Log. Comput.
  \textbf{31}(6),  1482--1525 (2021). \doi{10.1093/logcom/exab050}

\bibitem{DBLP:journals/tcs/Rutten00}
Rutten, J.J.M.M.: Universal coalgebra: a theory of systems. Theor.\ Comput.\
  Sci.  \textbf{249}(1),  3--80 (2000). \doi{10.1016/S0304-3975(00)00056-6}

\bibitem{DBLP:journals/tcs/Schroder08}
Schr{\"{o}}der, L.: Expressivity of coalgebraic modal logic: The limits and
  beyond. Theor.\ Comput.\ Sci.  \textbf{390}(2-3),  230--247 (2008).
  \doi{10.1016/j.tcs.2007.09.023}

\bibitem{DBLP:conf/fsttcs/SilvaBBR10}
Silva, A., Bonchi, F., Bonsangue, M.M., Rutten, J.J.M.M.: Generalizing the
  powerset construction, coalgebraically. In: Lodaya, K., Mahajan, M. (eds.)
  Foundations of Software Technology and Theoretical Computer Science, {FSTTCS}
  2010. LIPIcs, vol.~8, pp. 272--283. Schloss Dagstuhl -- Leibniz-Zentrum
  f{\"{u}}r Informatik (2010). \doi{10.4230/LIPIcs.FSTTCS.2010.272}

\bibitem{Smirnov08}
Smirnov, A.: Graded monads and rings of polynomials. J.\ Math.\ Sci.
  \textbf{151},  3032--3051 (2008)

\bibitem{DBLP:journals/lmcs/SokolovaW18}
Sokolova, A., Woracek, H.: Termination in convex sets of distributions. Log.
  Methods Comput. Sci.  \textbf{14}(4) (2018). \doi{10.23638/LMCS-14(4:17)2018}

\bibitem{DBLP:conf/calco/UrabeH15}
Urabe, N., Hasuo, I.: Coalgebraic infinite traces and kleisli simulations. In:
  Moss, L.S., Sobocinski, P. (eds.) Algebra and Coalgebra in Computer Science,
  {CALCO} 2015. LIPIcs, vol.~35, pp. 320--335. Schloss Dagstuhl --
  Leibniz-Zentrum f{\"{u}}r Informatik (2015).
  \doi{10.4230/LIPIcs.CALCO.2015.320}

\bibitem{bw:behavioural-pseudometric}
{van Breugel}, F., Worrell, J.: A behavioural pseudometric for probabilistic
  transition systems. Theoretical Computer Science  \textbf{331},  115--142
  (2005)

\bibitem{DBLP:journals/lmcs/WildS22}
Wild, P., Schr{\"{o}}der, L.: Characteristic logics for behavioural hemimetrics
  via fuzzy lax extensions. Log.\ Methods Comput.\ Sci.  \textbf{18}(2) (2022).
  \doi{10.46298/lmcs-18(2:19)2022}

\bibitem{Worrell00}
Worrell, J.: Coinduction for recursive data types: partial orders, metric
  spaces and omega-categories. In: Reichel, H. (ed.) Coalgebraic Methods in
  Computer Science, {CMCS} 2000. ENTCS, vol.~33, pp. 337--356. Elsevier (2000).
  \doi{10.1016/S1571-0661(05)80356-1}

\end{thebibliography}

\appendix

\end{document}